\documentclass[12pt]{article}
\usepackage{xspace}
\usepackage[colorlinks=true,linkcolor=magenta,citecolor=blue]{hyperref}
\usepackage{amssymb,amsmath,amsthm,graphicx,color,bbm}
\usepackage{thmtools, thm-restate}
\usepackage[margin=1.in]{geometry}
\usepackage{multirow,array}
\allowdisplaybreaks

\usepackage[capitalise,nameinlink]{cleveref}

\usepackage{algorithm}
\usepackage[noend]{algpseudocode}

\makeatletter
\newcommand*\rel@kern[1]{\kern#1\dimexpr\macc@kerna}
\newcommand*\widebar[1]{%
  \begingroup
  \def\mathaccent##1##2{%
    \rel@kern{0.8}%
    \overline{\rel@kern{-0.8}\macc@nucleus\rel@kern{0.2}}%
    \rel@kern{-0.2}%
  }%
  \macc@depth\@ne
  \let\math@bgroup\@empty \let\math@egroup\macc@set@skewchar
  \mathsurround\z@ \frozen@everymath{\mathgroup\macc@group\relax}%
  \macc@set@skewchar\relax
  \let\mathaccentV\macc@nested@a
  \macc@nested@a\relax111{#1}%
  \endgroup
}
\makeatother

\theoremstyle{plain}
\newtheorem{theorem}{Theorem}[section]
\newtheorem{corollary}[theorem]{Corollary}

\newtheorem{lemma}[theorem]{Lemma}
\newtheorem{claim}[theorem]{Claim}

\newtheorem{conjecture}[theorem]{Conjecture}

\newtheorem*{lemmaNoNum}{Lemma}

\newtheorem{definition}[theorem]{Definition}
\newtheorem{example}[theorem]{Example}

\theoremstyle{remark}

\theoremstyle{plain}

\def\R{{\mathbb{R}}}

\def\N{{\mathbb{N}}}

\renewcommand{\Pr}{\mathop{\bf Pr\/}}
\newcommand{\E}{\mathop{\bf E\/}}

\def\var{{\mathop{\bf Var\/}}}

\def\pmone{{\{\pm1\}}}
\def\sgn{\mathrm{sgn}}

\def\poly{{\mathrm{poly}}}

\renewcommand{\l}{\ell}

\newcommand{\remove}[1]{}

\newcommand{\mathify}[1]{\ifmmode{#1}\else\mbox{$#1$}\fi}

\newcommand{\complexityclass}[1]{{\bf{#1}}\xspace}

\newcommand{\BQP}{\complexityclass{BQP}}

\newcommand{\ACzero}{\complexityclass{AC^0}}

\newcommand{\eps}{\varepsilon}

\newcommand{\NN}{\mathcal{N}}
\newcommand{\cD}{{\mathcal{D}}}
\newcommand{\cDUk}{{\mathcal{D}_{U,k}}}

\newcommand{\cU}{{\mathcal{U}}}
\newcommand{\U}{{U}}
\newcommand{\T}{\mathsf{T}}

\newcommand{\polylog}{\mathrm{polylog}}

\renewcommand{\hat}{\widehat}
\renewcommand{\tilde}{\widetilde}
\renewcommand{\next}{\mathbf{Next}}
\newcommand{\ket}[1]{{\left\vert #1 \right\rangle}}

\newcommand{\ignore}[1]{}

\begin{document}

\title{Towards Optimal Separations between Quantum and Randomized Query Complexities}

\author{Avishay Tal\thanks{\texttt{atal@berkeley.edu}. Department of Electrical Engineering and Computer Sciences, University of California at Berkeley.}}
\maketitle

\begin{abstract}
The query model offers a concrete setting where quantum algorithms are provably superior to randomized algorithms. Beautiful results by Bernstein-Vazirani, Simon,  Aaronson, and others presented partial Boolean functions that can be computed by quantum algorithms making much fewer queries compared to their randomized analogs. To date, separations of $O(1)$ vs. $\sqrt{N}$ between quantum and randomized query complexities remain the state-of-the-art (where $N$ is the input length), leaving open the question of whether $O(1)$ vs. $N^{1/2+\Omega(1)}$ separations are possible?

We answer this question in the affirmative. Our separating problem is a variant of the Aaronson-Ambainis $k$-fold Forrelation problem. We show that our variant:
\begin{enumerate}
\item Can be solved by a quantum algorithm making $2^{O(k)}$ queries to the inputs.
\item Requires at least $\tilde{\Omega}(N^{2(k-1)/(3k-1)})$ queries for any randomized algorithm.
\end{enumerate}
For any constant $\varepsilon>0$, this gives a  $O(1)$  vs. $N^{2/3-\varepsilon}$ separation between the quantum and randomized query complexities of partial Boolean functions. 

Our proof is Fourier analytical and uses new bounds on the Fourier spectrum of classical decision trees, which could be of independent interest. 

Looking forward, we conjecture that the Fourier bounds could be further improved in a precise manner, and show that such conjectured bounds imply optimal $O(1)$ vs. $N^{1-\varepsilon}$ separations between the quantum and randomized query complexities of partial Boolean functions.
\end{abstract}

\thispagestyle{empty}

\newpage

\clearpage 
\thispagestyle{empty}
\tableofcontents
\newpage

\clearpage 

\setcounter{page}{1}
\section{Introduction}
The query model (or black-box model) offers a concrete setting where quantum algorithms are provably superior to their randomized counterparts. 
Many well-known quantum algorithms can be cast in this model, such as  Grover's search~\cite{Grover96}, Deutsch-Jozsa's algorithm~\cite{DJ92}, Bernstein-Vazirani's algorithm~\cite{BV97}, Simon's algorithm~\cite{Simon97}, and Shor's period-finding algorithm (which is a major component in Shor's factoring algorithm \cite{Shor94}).
In the query model, we seek to answer a question about the input by making as few queries to it as possible. For deterministic algorithms, this is also known as the decision tree model, where the decision tree depth equals the number of queries. The randomized and quantum versions of this model are very well-studied with many known connections and separations between the models in different settings (cf. the wonderful survey of \cite{BuhrmanW02} or the more recent work of~\cite{ABK16}).

A beautiful line of work showed that for partial Boolean functions on $N$ variables, the quantum query complexity could be exponentially smaller (or even less) than the randomized query complexity.
Separations of $O(\log N)$ vs.\ $\sqrt{N}$ date back to the work of Simon~\cite{Simon97} and similarly for a problem introduced by Childs et al.~\cite{ChildsCDFGS03}.
In \cite{Aaronson10,AA15}, it was shown that the problem of {\em Forrelation} exhibits a $1$ vs.\ $\tilde{\Omega}(\sqrt{N})$ separation between the quantum and randomized query complexities.

Buhrman et al. \cite{BFNR08} and Aaronson and Ambainis \cite{AA15} asked what are the best possible separations between quantum and randomized query complexities? Aaronson and Ambainis presented this question as a fundamental question whose answer will shed light on the differences between the two models and gave several results towards its answer. 

On the one hand, they ruled out $O(1)$ vs.\ $\Omega(N)$ separations. More precisely, they showed that for any constant $t$, any quantum algorithm that makes $t$ queries can be simulated (up to small error) by a randomized algorithm making  $\tilde{O}(N^{1-1/2t})$ non-adaptive queries.  For $t=1$, this shows that Forrelation is an extremal separation.

On the other hand, towards obtaining optimal $t$ vs.\ $\Omega(N^{1-1/2t})$ separations, they suggested a candidate: the $k$-fold Forrelation problem (to be defined shortly), where $k=2t$. They showed that a quantum algorithm making $\lceil{k/2\rceil} = t$ queries can solve $k$-fold Forrelation. Moreover, they conjectured that any randomized algorithm would require $\Omega(N^{1-1/k}) = \Omega(N^{1-1/2t})$ queries. While Aaronson and Ambainis proved the case $k=2$ in their conjecture, they left all other cases wide open. For $k> 2$, the lower bounds they obtained on $k$-fold Forrelation are of the form $\Omega(\sqrt{N}/(\log N)^{7/2})$.

Aaronson and Ambainis further noted that in all of the above exponential separations, the randomized query complexity did not surpass $\sqrt{N}$. They asked whether separations of $\polylog(N)$ vs.\ $N^{1/2+\Omega(1)}$ are possible?

We answer their question in the affirmative. We revisit their candidate, changing it in a way that would be crucial for our analysis. First, we define the {\bf Rorrelation} of $k$ vectors, with respect to an $N$-by-$N$ orthogonal matrix $\U$.
\begin{definition}
Let $\U\in \R^{N\times N}$ be an orthogonal matrix. The {\bf $k$-fold Rorrelation} of vectors $z^{(1)}, \ldots, z^{(k)}\in \R^N$ with respect to $\U$ is defined as
$$
	\phi_{\U}(z^{(1)}, \ldots, z^{(k)}) = \frac{1}{N}\cdot \sum_{i_1=1}^{n} \ldots \sum_{i_k=1}^{n} z^{(1)}_{i_1} \cdot \U_{i_1,i_2} \cdot z^{(2)}_{i_2} \cdot  \U_{i_2,i_3}  \cdots \U_{i_{k-1},i_k} \cdot z^{(k)}_{i_k}\;.
$$
\end{definition}
One can pick $\U$ to be the $N$-by-$N$ Hadamard matrix, as suggested by Aaronson and Ambainis. 
We, however, pick $\U$ uniformly at random from all $N$-by-$N$ orthogonal matrices.%
\footnote{Aaronson Ambainis called their variant Forrelation, as in the case $k=2$ it measures {\bf correlation} after applying the {\bf F}ourier/Hadamard transform on $z^{(1)}$. We measure correlation after applying a {\bf R}andomly chosen orthogonal matrix, hence the name {\bf Rorrelation}.}
This will play a major role later on, since we rely on properties that hold with high probability for a random orthogonal matrix, but do not hold for the Hadamard matrix (see Def.~\ref{def:good}).

It is not hard to show that the $k$-fold Rorrelation of any vectors $z^{(1)}, \ldots, z^{(k)} \in \{-1,1\}^N$ is at most $1$ in absolute value.
The computational task we consider in this paper, called the {\bf $k$-fold Rorrelation Problem}, asks to distinguish between the following two cases:\begin{description}
\item[YES Instances:]  vectors $z^{(1)}, \ldots, z^{(k)} \in \{-1,1\}^N$ with 
$\phi_{\U}(z^{(1)}, \ldots, z^{(k)})\ge 2^{-k}$, and  
\item[NO Instances:] vectors $z^{(1)}, \ldots, z^{(k)} \in \{-1,1\}^N$ with 
$|\phi_{\U}(z^{(1)}, \ldots, z^{(k)})|\le \frac{1}{2} \cdot 2^{-k}$.
\end{description}
We shall show that while the $k$-fold Rorrelation problem is easy in the quantum query model (for any choice of $\U$), it requires many queries in the classical setting (for most choices of $\U$).
Namely, our main separation will show that:
\begin{enumerate}
	\item For any $N$-by-$N$ orthogonal matrix $U$, there exists a quantum algorithm making at most $2^{O(k)}$ queries that solves the $k$-fold Rorrelation problem (with respect to $U$) with success probability at least $2/3$.
	\item  For most $N$-by-$N$ orthogonal matrices $U$, any randomized algorithm   that solves the $k$-fold Rorrelation problem (with respect to $U$)  with success probability at least $2/3$, must make at least $\Omega(N^{2(k-1)/(3k-1)}/k \log N)$ queries to the inputs.
\end{enumerate}
So far, we left the choice for the value of $k$ to be arbitrary. We think of $k$ as either a fixed constant or a slow-growing function of $N$, in particular $k = o(\log N)$.
By picking $k$ to be a large constant, the above discussion gives a $O(1)$ vs.\ $\Omega(N^{2/3-\eps})$ separation of quantum and randomized query complexities, for any small constant $\eps>0$.
By picking $k=O(\log \log N)$, we get a $O(\log N)$ vs.\ $N^{2/3-o(1)}$ separation of the two measures.

Before explaining more about our techniques and the potential room for improvement, we describe an application of our results to another setting.

\paragraph{Application: Power-$(2+2/3)$ Separations for Total Boolean Functions.}
While for partial Boolean functions (or promise problems) exponential separations are possible between randomized and quantum query complexities, for total functions (i.e., Boolean functions that are defined on the entire domain $\{-1,1\}^N$) the picture is quite different. The seminal work of Beals et al. \cite{BBCMW01} showed (among others) that quantum query complexity and randomized query complexity are at most power-$6$ apart. That is, $R(f) \le O(Q(f)^6)$ for any total Boolean function $f$, where $R(f)$ and $Q(f)$ denote the randomized and quantum query complexities of $f$, respectively.

On the other hand, Grover's search demonstrated that for the OR function $R(f) \ge \Omega(Q(f)^2)$ \cite{Grover96}, and this is tight~\cite{BBBV97}.
Two decades later, Aaronson, Ben-David, Kothari~\cite{ABK16} exhibited the first super-quadratic separations between $Q(f)$ and $R(f)$ for total functions, surprisingly improving Grover's separation that was believed to be optimal.
Their work presented a power-2.5 separation based on the ``cheat-sheet'' technique applied to $2$-fold Forrelation.
More generally, they showed that any $N^{o(1)}$ vs.\ $N^{c-o(1)}$ separation between the quantum and randomized query complexities of partial functions, implies a power-$(2+c)$ separation for total Boolean functions.
Plugging in our result, yields a power-$(2+2/3)$ separation for total Boolean functions.
In other words, the transformation of \cite{ABK16} applied to $k$-fold Rorrelation yields a total function  $f_{\mathrm{CS}}$ such that $R(f_{\mathrm{CS}}) \ge Q(f_{\mathrm{CS}})^{2+\frac{2}{3}-o(1)}$.

\subsection{Our Techniques}

\subsection*{Quantum Query Complexity of the $k$-fold Rorrelation Problem.}
A simple adaptation of the algorithm suggested by Aaronson and Ambainis \cite{AA15} shows the existence of a $\lceil{k/2\rceil}$-query quantum algorithm on inputs $z^{(1)}, \ldots, z^{(k)} \in \{-1,1\}^N$, whose acceptance probability equals 
$$\frac{1+\phi_{\U}(z^{(1)}, \ldots, z^{(k)})}{2}$$
This shows that there's a gap of $\frac{1+2^{-k}}{2}$ vs.\ $\frac{1+2^{-(k+1)}}{2}$ between the acceptance probabilities in the YES instances and NO instances.
For $k$ a fixed constant this gives a constant difference between the acceptance probabilities of the YES and NO instances.
If $k = \omega(1)$ or if we insist on getting a $2/3$ vs.\ $1/3$ separation between the acceptance probabilities of YES and NO instances, then one apply simple amplification techniques repeating the quantum algorithm for $2^{O(k)}$ times, and check whether the number of accepting trials exceeds a certain threshold.	

\subsection*{Randomized Query Complexity of the $k$-fold Rorrelation Problem.}
Towards showing that the randomized query complexity of the $k$-fold Rorrelation problem is large, we construct a hard-distribution, and show that it is hard to solve the Rorrelation problem on instances sampled from this distribution.
By Yao's minimax principle, it suffices to show that a deterministic decision tree cannot solve the Rorrelation problem on the hard-distribution with high probability.
We take the hard-distribution to be the convex combination (i.e., average) of two distributions: (1) the uniform distribution on $k$ vectors $z^{(1)}, \ldots, z^{(k)}\in \{-1,1\}^N$, denoted $\cU_k$, and (2) 
a distribution $\cDUk$ that often produces $k$ vectors with large $k$-fold Rorrelation.
On the one hand, we show that $\cU_k$ produces NO instances with very high probability whereas $\cDUk$ produces YES instances with not too small probability.%
\footnote{We believe that $\cDUk$ samples YES instances with very high probability, but since this seems technically involved, and since it is not required to complete the proof, we left this question open.}
On the other hand, we show that any depth-$d$ deterministic decision tree fails to distinguish between $\cU_k$ and $\cDUk$, as long as $d = o(N^{2(k-1)/3(k-1)}/\log N)$. Combining the two facts together, we deduce that the distribution $\frac{1}{2} \cDUk + \frac{1}{2} \cU_k$ is a hard-distribution for depth $d$ decision trees. That is, any depth-$d$ decision tree errs with not too small probability in computing the Rorrelation problem on instances sampled from this distribution. 

We view the construction of a hard-distribution as an important contribution to the project set up by Aaronson and Ambainis. They were able to analyze $2$-fold Forrelation by presenting a hard-distribution for that case, but no candidate hard-distribution was suggested for the case $k>2$ prior to this work. We believe that our distribution is hard even for decision trees of depth $N^{1-1/k}/\polylog(N)$ and pose a conjecture that would imply such a result.

By the above discussion, proving that $\frac{1}{2} \cDUk + \frac{1}{2} \cU_k$ 
is a hard-distribution boils down to showing that:
\begin{enumerate}
	\item $\cU_k$ samples NO-instances with very high probability.
	\item $\cDUk$ samples YES-instances with not too small probability (to be precise, at least $2^{-k}$).
	\item Any deterministic decision tree of depth $d = o(N^{2(k-1)/3(k-1)}/\log N)$  cannot distinguish between inputs sampled from $\cU_k$ and inputs sampled from $\cDUk$. Put differently, the acceptance probability of any such deterministic decision tree, will be the same in both cases, up to an additive small error $o(1/2^k)$.
\end{enumerate}
Item 1 holds regardless of the choice of $\cDUk$, and is simple to prove.
To obtain Item 2, we start by recalling the hard distribution that Aaronson and Ambainis suggested for the case $k=2$.
They first defined a multi-variate Gaussian distribution $G_2$ on $2N$ dimensions, where
the first $N$ variables are simply standard independent Gaussians, and the latter $N$ variables are obtained by applying the Fourier (or Hadamard) transform on the first $N$ variables.
Then, to get a distribution $\cD_2$ over the Boolean domain, they took the signs of these multi-variate Gaussians. They then show that the expected Forrelation value of vectors sampled from $\cD_2$ is at least $(2/\pi)$.

We generalize this hard distribution to $k$-fold Rorrelation, replacing the Hadamard matrix with the orthogonal matrix $U$, and handling arbitrary $k\in \N$ rather then just $k=2$.

\subsection*{The Distributions $G_k$ and $\cDUk$}
Let $N,k\in \N$.
First, we define a continuous distribution $G_k$ over $\R^{kN}$ (in which every coordinate will be either a Gaussian random variable or a product of two independent Gaussian random variables), and then derive from it a discrete distribution over $\{-1,1\}^{kN}$ by taking signs.

The definition of $G_k$ and $\cDUk$ will rely on the $N$-by-$N$ orthogonal matrix $\U$ from the definition of Rorrelation.
For $i=1, \ldots, k-1$ let $X^{(i)} \sim \NN(0,1)^N$ and sample all the vectors $X^{(1)}, \ldots, X^{(k-1)}$ independently.
Denote by $Y^{(1)}, \ldots, Y^{(k-1)}$ the vectors defined by $Y^{(i)} = U^{\T} \cdot X^{(i)}$.
Define $Z^{(1)}, \ldots, Z^{(k)}$ as follows:
\begin{enumerate}
	\item  $Z^{(1)} = X^{(1)}$
	\item For $i=2, \ldots, k-1$, let  $Z^{(i)} = Y^{(i-1)} \odot X^{(i)}$ (where $\odot$ denotes point-wise product of two vectors of length $N$).
	\item $Z^{(k)} = Y^{(k-1)}$.
\end{enumerate}
We denote by $Z = (Z^{(1)}, \ldots, Z^{(k)})$ the concatenation of all the $k$ vectors.  This defines the distribution $G_k$ over $\R^{kN}$. The distribution $\cDUk$ over $\pmone^{kN}$ will simply be the distribution of $\sgn(Z)=(\sgn(Z^{(1)}_1), \sgn(Z^{(1)}_2), \ldots, \sgn(Z^{(k)}_{N}))$.

\subsection*{$\cDUk$ produces vectors with large Rorrelation:}
In section~\ref{sec:cD_k}, we show that $\E_{z\sim \cDUk}[\phi_{\U}(z)] \ge (2/\pi)^{k-1}$. Based on that, a simple Markov's inequality shows that with probability at least $2^{-k}$, $\cDUk$ samples a YES-instance for the Rorrelation problem. 
This will complete Item~2 in the aforementioned scheme, and we are left to prove the third item, which is the core of our argument.

\subsection*{The Core of the Argument: $\cDUk$ is Pseudorandom against Shallow  Decision Trees}

We are left to prove that for any depth $d = o(N^{2(k-1)/3(k-1)}/\log N)$  decision tree $F$, we have
$$
\left|\E_{z\sim \cU_k}[F(z)]-\E_{z\sim \cDUk}[F(z)]\right| \le o(1/2^k)
$$
We call $\left|\E_{z\sim \cU_k}[F(z)]-\E_{z\sim \cDUk}[F(z)]\right|$ the {\bf advantage} of $F$ distinguishing between $\cU_k$ and $\cDUk$.
Intuitively, a small advantage means that $F$ behaves similarly in both cases.
To bound the advantage of $F$, we apply a straight-forward Fourier analytical approach, as follows.
Since $F$ is defined over the Boolean domain, it can be represented as a multilinear polynomial, also known as the Fourier transform.
That is, we may write 
$$
F(z) = \sum_{S \subseteq \{1,\ldots, kN\}} \hat{F}(S) \cdot \prod_{i\in S} z_i.
$$
where $\hat{F}(S)$ are the real-valued Fourier coefficients of $F$. 
Observe that $\E_{z\sim \cU_k} [F(z)] = \hat{F}(\emptyset)$, whereas 
$$\E_{z\sim \cDUk} [F(z)] = \sum_{S \subseteq \{1,\ldots, kN\}} \hat{F}(S) \cdot \E_{z\sim \cDUk} \left[\prod_{i\in S} z_i\right].$$
To make our notation shorter, we denote by $\hat{\cDUk}(S) = \E_{z\sim \cDUk} \left[\prod_{i\in S} z_i\right]$. 
Thus, the advantage of $F$ can be expressed as
\begin{equation*}
	\left|\E_{z\sim \cU_k} [F(z)]- \E_{z\sim \cDUk} [F(z)]\right| = \left|\sum_{S \subseteq \{1,\ldots, kN\}: S\neq \emptyset} \hat{F}(S) \cdot \hat{\cDUk}(S)\right|,
\end{equation*}
which by the triangle inequality is at most
\begin{equation}\label{eq:adv}
	\sum_{S \subseteq \{1,\ldots, kN\}: S\neq \emptyset} 
	\left|\hat{F}(S) \cdot \hat{\cDUk}(S)\right|.
\end{equation}
We bound the sum in Eq.~\eqref{eq:adv}, by accounting for each degree $\ell\in [k N]$,d the contribution of sets of size $\ell$ to the sum.
Namely, 
 $$
\sum_{S \subseteq \{1,\ldots, kN\}: S\neq \emptyset} \left|\hat{F}(S) \cdot \hat{\cDUk}(S)\right| = \sum_{\ell=1}^{kN} \sum_{S:|S|=\ell} \left|\hat{F}(S) \cdot \hat{\cDUk}(S)\right|.
$$
To bound each internal sum $\sum_{S:|S|=\ell} \left|\hat{F}(S) \cdot \hat{\cDUk}(S)\right|$ it suffices to show that:
\begin{enumerate}
	\item 
	For every set $S\subseteq \{1, \ldots, kN\}$ of size $\ell$, the coefficient $|\hat{\cDUk}(S)|$ is sufficiently small.
	\item 	
	The sum $\sum_{S:|S|=\ell} |\hat{F}(S)|$ is not too large.
\end{enumerate}
This suffices to bound $\sum_{S:|S|=\ell} \left|\hat{F}(S) \cdot \hat{\cDUk}(S)\right|$, by the following the simple inequality
$$\sum_{S:|S|=\ell} \left|\hat{F}(S) \cdot \hat{\cDUk}(S)\right| \le  \left( \sum_{S:|S|=\ell}|\hat{F}(S)|\right) \cdot \left(\max_{S:|S|=\ell} {|\hat{\cDUk}(S)|}\right).$$
The proofs of both Parts~1 and~2 in the above plan are technically involved.

\paragraph{Part 1: Moment Bounds on $\cDUk$.}
Part~1 boils down to showing that all moments of the distribution $\cDUk$ are sufficiently small, where the bounds improve as the degree increases.
This is where the properties of random orthogonal matrices play their role. In particular we are able to show the following bound. 
\begin{theorem}
	With high probability over the choice of a uniformly random  orthogonal matrix $U\in \R^{N \times N}$, for all $\ell\in \{1,\ldots, kN\}$ and all sets $S \subseteq \{1, \ldots, kN\}$ of size $\ell$
	it holds that 
	\begin{equation}\label{eq:moment bounds}
			|\hat{\cDUk}(S)| \le \left(\frac{c \cdot \ell \cdot \log N}{N}\right)^{\ell\cdot (1-1/k)/2}
	\end{equation}
	where $c$ is some universal constant.
	Moreover, $\hat{\cDUk}(S)=0$ for all non-empty sets $S$ of size less than $k$.
\end{theorem}
For example for $k=\ell=2$ the theorem shows that any second moment of $\cD_{U,2}$ is at most $\tilde{O}(1/\sqrt{N})$.
For any constant $k$ and $\ell=k$, the theorem shows that the $k$-th moment of $\cDUk$ is at most $\tilde{O}(1/\sqrt{N^{k-1}})$.

We remark that replacing $U$ with the Hadamard matrix, one gets much worse bounds for high moments, 
that would not allow to prove better than $\sqrt{N}$ lower bounds on the decision tree depth using our approach.
Furthermore, we believe our bounds on $|\hat{\cDUk}(S)|$ are tight. 

\paragraph{Part 2: Level-$\ell$ Fourier Bounds on $F$.}
We return to Part~2 above, bounding the sum $\sum_{S:|S|=\ell} |\hat{F}(S) |$ where $F$ is a decision tree of depth $d$. The work of O'Donnell and Servedio \cite{ODonnellServedio:07} obtained a tight $O(\sqrt{d})$ bound for the case $\ell=1$ (which was later extended by Blais, Tan, and Wan \cite{BTW15} to parity decision trees).
This allowed O'Donnell and Servedio to obtain a polynomial-time algorithm for learning monotone decision trees under the uniform distribution.
To the best of our knowledge, the question about higher Fourier levels of decision trees was not explicitly raised in the literature prior to this work.

We denote by $L_{1,\ell}(F) = \sum_{S:|S|=\ell} |\hat{F}(S)|$.
Bounds of the quantity $L_{1,\ell}(F)$ were proved for several well-studied classes of Boolean functions such as bounded-width DNF formulas \cite{Man95b}, bounded depth circuits \cite{Tal17}, read-once branching programs \cite{ReingoldSV13,CHRT18}, functions with low sensitivity \cite{GSTW16}, and low-degree polynomials over finite fields \cite{CHHL18}. Furthermore, it was conjectured in \cite{CHLT19} that stronger classes of Boolean functions,  such as $\ACzero[\oplus]$ circuits, have low $L_{1,\ell}(F)$. Moreover, the work of \cite{CHHL18} showed how to generically construct  pseudorandom generators assuming only bounds on the $L_{1,\ell}(F)$ of functions in the family (or even assuming only bounds on $L_{1,2}(F)$ \cite{CHLT19}). 

The quantity $L_{1,\ell}(F)$ captures the ``sparsity'' of the Fourier spectrum. Intuitively, this is due to the known fact that the sum of squares $\sum_{S:|S|=\ell} |\hat{F}(S)|^2$ is at most $1$. Hence, having the sum of absolute values small, means that most of the Fourier mass sits on a few coefficients.

We prove new bounds on the $L_{1,\ell}(F)$ of any decision tree $F$ of depth $d$.
\begin{theorem}
	Let $F$ be a decision tree on $kN$ input variables of depth $d$.
	Then,
\begin{equation}
	\label{eq:sparsity bounds}
	\forall{\ell}:\sum_{S\subseteq \{1, \ldots, kN\}:|S|=\ell}{|\hat{F}(S)|} \le \sqrt{d^{\ell} \cdot O(\log kN)^{\ell-1}}
\end{equation}
\end{theorem}
In particular, for $\ell=1$, we match the tight $O(\sqrt{d})$ bound of \cite{ODonnellServedio:07}. 
Moreover, for constant $\ell$,  
our bounds are nearly tight as exhibited by the composition of the Address and the Majority functions (see Section~\ref{sec:examples}). However, for higher values of $\ell$, our bounds get sloppier and we believe that they can be further improved as follows.
\begin{conjecture}\label{conj}
		Let $F$ be a decision tree on $kN$ input variables of depth $d$.
	Then,
	$$
	\forall{\ell}:\sum_{S\subseteq \{1, \ldots, kN\}:|S|=\ell}{|\hat{F}(S)|} \le \sqrt{\binom{d}{\ell} \cdot O(\log kN)^{\ell-1}}
	$$
\end{conjecture}
We remark that for non-adaptive decision trees, namely juntas, the conjecture  holds.\footnote{If $F$ is a Boolean $d$-junta, then there are at most $\binom{d}{\ell}$ non-zero Fourier coefficients of degree $\ell$, thus by Cauchy-Schwarz
$
	\sum_{S:|S|=\ell}{|\hat{F}(S)|} \le \sqrt{\binom{d}{\ell} \cdot \sum_{S:|S|=\ell}{|\hat{F}(S)|^2}} \le \sqrt{\binom{d}{\ell}}.$}
Combining the bounds in Eq.~\eqref{eq:moment bounds} with Eq.~\eqref{eq:sparsity bounds} gives 
$$
\left( \sum_{S:|S|=\ell}|\hat{F}(S)|\right) \cdot \left(\max_{S:|S|=\ell} {|\hat{\cDUk}(S)|}\right) \le o(1/2^\ell)$$
for all $\ell \le \sqrt{d/\log n}$ and all $d = o(N^{2(k-1)/3(k-1)}/\log N)$.
For larger degrees $\ell > \sqrt{d/\log n}$, we use a much simpler bound, $L_{1,\ell}(F) \le \binom{d}{\ell}$, to obtain similarly 
$$\left( \sum_{S:|S|=\ell}|\hat{F}(S)|\right) \cdot \left(\max_{S:|S|=\ell} {|\hat{\cDUk}(S)|}\right) \le o(1/2^\ell).$$
Summing over all sets of size at least $k$ we get
$$
\left|\E_{z\sim \cU_k} [F(z)]- \E_{z\sim \cDUk} [F(z)]\right| \le \sum_{\ell=k}^{kN} \left( \sum_{S:|S|=\ell}|\hat{F}(S)|\right) \cdot \left(\max_{S:|S|=\ell} {|\hat{\cDUk}(S)|}\right) \le o(1/2^k),
$$
which completes the proof. 

We remark that assuming Conjecture~\ref{conj}, the same strategy would work for any decision tree of depth at most $N^{1-1/k}/\polylog(N)$.

\subsection{Digest - Degree and Sparsity}
We seek to pinpoint the key differences between the quantum and randomized query models that allowed us to get this separation.

The seminal result of Beals et al.~\cite{BBCMW01} showed that the acceptance probability of any $t$ query quantum algorithm can be expressed as a degree-$2t$ multilinear polynomial. Thus quantum algorithms making few queries can be approximated by low-degree polynomials. 
This is also the case for randomized decision trees, as observed by \cite{NisanSzegedy94}. 
But, if both models are approximated by low-degree polynomials, what is the difference between them?

We suggest sparsity, or more precisely $L_{1,\ell}(\cdot)$, as a measure to separate the two models. This is evident from our proof, which strongly exploits the smallness of $L_{1,\ell}(F)$ for shallow randomized decision trees. 
On the other hand, one can show that the $L_{1,\ell}(\cdot)$ of quantum algorithms making a few queries could be quite large.
As mentioned above, for any quantum algorithm making $t$ queries, there exists a multilinear polynomial $p$ of degree $2t$ capturing the acceptance probability of the algorithm. 
With this in mind, one can analyze the $L_{1,\ell}(\cdot)$ of $p$, i.e., the sum of absolute values of the degree $\ell$ terms in the polynomial $p$. 
Indeed, the polynomial that arises from Aaronson-Ambainis algorithm (Claim~\ref{claim:AA for Rorrelation}) is exactly 
$\frac{1}{2}(1+\phi_{\U}(z^{(1)}, \ldots, z^{(k)}))$.
Observe that for a random orthogonal matrices $U$,  entries in the matrix $U$ are of magnitude roughly $1/\sqrt{N}$, with high probability, and thus the sum of absolute values of the  degree-$k$ coefficients in $\phi_{\U}$ is quite large, $\tilde{\Theta}(N^{(k-1)/2})$. 

This captures the difference between the models. Indeed, to get such large $L_{1,k}$ measure for randomized decision tree, one needs their depth to be at least $\tilde{\Omega}(N^{(1-1/k)})$. 

We remark however that differences in $L_{1,\ell}(\cdot)$ alone are not sufficient to show a difference between the computational abilities of the two models. 
Indeed, two polynomials with very similar values on the entire Boolean domain can get very different $L_{1,\ell}$ norms.
This is why it is non-trivial to find and prove that a computational task demonstrates these differences. 
As we show in this paper, the $k$-Fold Rorrelation problem is such a task.

\subsection{Related Work}
We would like to comment about the relation of this work with our prior joint work with Raz~\cite{RT19}. 
In~\cite{RT19}, a separation between quantum algorithms, making a few queries,  and $\ACzero$ circuits was obtained. This, in turn, implied an oracle separation between $\BQP$ and the Polynomial Hierarchy.  
The question in \cite{RT19} boiled down to whether a distribution similar%
\footnote{A major difference, though, is that when discretizing the Gaussian distribution, \cite{RT19} used randomized rounding whereas we take signs.} to $\cD_2$ is pseudorandom against $\ACzero$ circuits.
While the proof strategy in \cite{RT19} starts similarly to the one laid out here, it takes a sharp turn early on. 
Namely, there, the approach of bounding the contribution of each level separately, simply does not work. 
To overcome this hurdle, one needs to rely on techniques from stochastic calculus, viewing the Gaussian distribution as a result of a random walk that makes many tiny steps, and analyzing each step separately. 
Surprisingly, the result of \cite{RT19} relies only on bounds on the second-level of the Fourier spectrum of $\ACzero$ (i.e., only bounds on $L_{1,2}(F)$).

Here, we exploit delicate bounds on all levels of the Fourier spectrum of depth-$d$ decision trees (i.e., bounds on $L_{1,\ell}(F)$ for all $\ell$, where $F$ is a depth-$d$ decision tree), as well as tight moment bounds on the distribution $\cDUk$.
So far, despite initial attempts, we were unable to exploit techniques from stochastic calculus to analyze $\E_{z\sim \cDUk}[F(z)]$. One difficulty arises from the fact that the continuous distribution $G_k$, which we discretize to get $\cDUk$, involves products of  Gaussians, rather than just Gaussians.
It seems tempting to wonder whether only a bound on the $k$-th Fourier level would suffice to analyze $k$-fold Rorrelation. If so, this would give a completely different proof, with possibly optimal quantitative bounds.

\section{Preliminaries}
For $N\in \N$, we denote by $[N] = \{1,\ldots, N\}$. We denote by $I_N$ the identity matrix of order $N$.
For  $A \in \R^{m\times n}$, we denote by $\|A\|$ the matrix norm given by $\|A\| = \sup_{x\neq 0}\|A x\|_2/\|x\|_2$.

\subsection*{Quantum Query}
A quantum query to an input $z  \in \pmone^{kN}$ performs the diagonal unitary transformation ${U}_z$, defined by
$|i,w \rangle \rightarrow z_i|i,w \rangle,$
where $i \in [kN]$ and $w$ represents the auxiliary workspace that does not participate in the query.
\subsection*{Fourier Representation of Boolean Functions}
Let $f: \pmone^N \to \R$ be a Boolean function on $N$ variables. The Fourier transform of $f$ is the unique multilinear polynomial that agrees with $f$ on $\pmone^N$. Such a polynomial exists and is unique.
We write the Fourier transform as
$$
f(x) = \sum_{S \subseteq [N]} {\hat{f}(S) \cdot \prod_{i\in S} x_i}
$$
where $\hat{f}(S)\in \R$ are the Fourier-coefficients, that could be easily computed from the function $f$ by $\hat{f}(S) = \E_{x\in\pmone^N}[f(x) \cdot \prod_{i\in S} x_i]$.
Parseval's identity shows that $\E_{x\in \pmone^N}[f(x)^2] =\sum_{S \subseteq [N]} {\hat{f}(S)^2}$. For $\ell\in [N]$, we denote by 
$$
L_{1,\ell}(f) \triangleq \sum_{S\subseteq [N]:|S|=\ell} |\hat{f}(S)|.
$$
\subsection*{Moments of Distributions}
Let $D$ be a distribution over $\{-1,1\}^N$. 
For any subset $S \subseteq [N]$ we denote by $$\hat{D}(S) = \E_{x\sim D}\left[\prod_{i\in S}x_i\right].$$

\section{Quantum Algorithm for the $k$-fold Rorrelation Problem}
Aaronson and Ambainis \cite{AA15} presented an algorithm that solves Forrelation (the special case of Rorrelation when $\U$ is the Hadamard matrix) with only $\lceil{k/2\rceil}$ queries.
 It is straightforward to extend their algorithm to solve the Rorrelation problem.
 \begin{claim}\label{claim:AA for Rorrelation}
Let $N$ be a power of $2$. Let $\U$ be an $N$-by-$N$ orthogonal matrix.
 Then, there exists a quantum algorithm making $\lceil{k/2\rceil}$ quantum queries, whose acceptance probability is $\frac{1}{2}  + \frac{1}{2} \phi_{\U}(z^{(1)},\ldots, z^{(k)})$
 where recall that
\[ \phi_{\U}(z^{(1)},\ldots, z^{(k)}) = \frac{1}{N}\cdot \sum_{i_1, \ldots, i_k} z^{(1)}_{i_1} \cdot \\U_{i_1,i_2} \cdot z^{(2)}_{i_2} \cdot  \\U_{i_2,i_3}  \cdots \\U_{i_{k-1},i_k} \cdot z^{(k)}_{i_k}.
\]\end{claim}
 For completeness, we give the proof that naturally extends \cite[Prop.~6]{AA15} in Appendix~\ref{app:AA}.
 Note that Claim~\ref{claim:AA for Rorrelation} means that the adaptation of  Aaronson-Ambainis's algorithm accepts YES-instances with probability at least $\frac{1+2^{-k}}{2}$ and accepts NO-instances with probability at most $\frac{1+2^{-(k+1)}}{2}$. 

This can be amplified to a $2/3$ vs.~$1/3$ separation as explained next.
Denote by
 $\eps$ the difference of the two fractions $\frac{1+2^{-k}}{2} - \frac{1+2^{-(k+1)}}{2}$ and by $\alpha$ their average.
 By the standard amplification technique of repeating an algorithm for $m= O(1/\eps^2) = O(4^k)$ times and checking whether the number of successful trials exceeds $m\cdot \alpha$,  we strengthen the separation between the acceptance probabilities of YES and NO instances to $2/3$ vs.\ $1/3$.

\begin{corollary}\label{cor:AA amplified}
Let $N$ be a power of $2$. Let $\U$ be an $N$-by-$N$ orthogonal matrix. Then, there exists a quantum algorithm making $O(k \cdot 4^k)$ quantum queries, that solves $k$-fold Rorrelation problem with respect to $\U$ with success probability at least $2/3$.\end{corollary}
 
 We remark that while the algorithms mentioned in Claim~\ref{claim:AA for Rorrelation} or Corollary~\ref{cor:AA amplified} make only a few quantum queries, they are not necessarily efficient in terms of running time as they apply an arbitrary orthogonal transformation $\U$ to the quantum register.
It remains an important open problem to show that one can get similar separations for orthogonal matrices $\U$ that can be  implementable efficiently, say by quantum circuits with at most $\polylog(N)$ many gates.
This boils down to showing the existence of efficiently implementable matrices $\U$ that satisfy the pseudorandomness condition in Def.~\ref{def:good}.

\section{The Rorrelation of Vectors Sampled from $\cDUk$~and~$\cU_k$}
In this section, we show that:
\begin{enumerate}
	\item Vectors $z^{(1)}, \ldots, z^{(k)}$ that are sampled from the distribution $\cDUk$ have expected Rorrelation value at least $(2/\pi)^{k-1}$.
	\item Vectors $z^{(1)}, \ldots, z^{(k)}$ that are sampled from the uniform distribution $\cU_k$ have expected Rorrelation value $0$. Furthermore, the variance of their Rorrelation value is $1/N$, thus it is highly concentrated around $0$.
\end{enumerate}
This shows that the algorithm from Claim~\ref{claim:AA for Rorrelation} distinguishes between $\cDUk$ and $\cU_k$ as its acceptance probability differs by at least $\frac{1}{2} \cdot (2/\pi)^{k-1}$ between the two cases. 
Then, in Section~\ref{sec:lower bound} we show that bounded-depth randomized decision trees fail to distinguish between $\cDUk$ and $\cU_k$, for most orthogonal matrices $\U$. This, in turn, leads to the conclusion that bounded-depth randomized decision trees fail to solve the Rorrelation problem, for most orthogonal matrices $\U$.

\subsection{The Rorrelation of Vectors Sampled from $\cDUk$}\label{sec:cD_k}
\begin{claim}
Let $\U$ be an $N$-by-$N$ orthogonal matrix.
Then, 
$$\E_{z\sim \cDUk}[\phi_{\U}(z^{(1)},\ldots, z^{(k)})] \ge (2/\pi)^{k-1}.
	$$
\end{claim}
\begin{proof}
The expectation of $\phi_{\U}$ on 
the distribution $\cDUk$ is
\begin{align*}&\E_{z\sim \cDUk}[\phi_{\U}(z^{(1)},\ldots, z^{(k)})] =\E_{Z\sim G_k}[\phi_{\U}(\sgn(Z^{(1)}),\ldots, \sgn(Z^{(k)}))] \\&= \frac{1}{N}\cdot \sum_{i_1, \ldots, i_k} \E[\sgn(X^{(1)}_{i_1}) \cdot \U_{i_1,i_2} \cdot \sgn(Y^{(1)}_{i_2}) \cdot \sgn(X^{(2)}_{i_2}) \cdot  \U_{i_2,i_3}  \cdots  \U_{i_{k-1},i_k} \cdot \sgn(Y^{(k-1)}_{i_k})]\\&=
\frac{1}{N}\cdot \sum_{i_1, \ldots, i_k} 
\E[\sgn(X^{(1)}_{i_1}) \cdot \U_{i_1,i_2} \cdot \sgn(Y^{(1)}_{i_2})]  
\cdots \E[\sgn(X^{(k-1)}_{i_{k-1}}) \U_{i_{k-1},i_k} \cdot \sgn(Y^{(k-1)}_{i_k})].
\end{align*}
In the following Lemma~\ref{lemma:2.1}, we show that for any $j\in \{1,\ldots,k-1\}$ and any $i_j\in [N]$ and $i_{j+1}\in [N]$
the expectation of $$\E[\sgn(X_{i_j}^{(j)}) \cdot \U_{i_j, i_{j+1}} \cdot \sgn(Y_{i_{j+1}}^{(j+1)})] \ge \frac{2}{\pi} \cdot \U^2_{i_j,i_{j+1}},$$
relying on the fact that the covariance of $X_{i_j}^{(j)}$ and $Y_{i_{j+1}}^{(j+1)}$ equals $U_{i_j,i_{j+1}}$.
This gives
\begin{align*}\E_{z\sim \cDUk}[\phi_{\U}(z^{(1)},\ldots, z^{(k)})]  &\ge \frac{1}{N}\cdot \sum_{i_1, \ldots, i_k} \U_{i_1,i_2}^2 \cdots  \U_{i_{k-1},i_k}^2  \cdot \left(\frac{2}{\pi} \right)^{k-1}\\&= \frac{1}{N}\cdot N \cdot \left(\frac{2}{\pi} \right)^{k-1}  \tag{since $\U$ is orthogonal}\\&= \left(\frac{2}{\pi} \right)^{k-1} \;.\qedhere
\end{align*}
\end{proof}

\begin{lemma}\label{lemma:2.1}
Let $\rho \in [-1,1]$. Let $(X, Y)$ be two-dimensional multi-variate Gaussian distribution with zero-means and covariance matrix 
$\left(\begin{matrix}
	1&\rho\\
	\rho&1
\end{matrix}\right)$.
Then,
$$\rho \cdot \E[\sgn(X) \cdot \sgn(Y) ] \ge \frac{2}{\pi} \cdot \rho^2.$$
\end{lemma}
\begin{proof}
Let $u_1 = (1,0)$ and $u_2 = (\rho,\sqrt{1-\rho^2})$.
We can retrieve such a distribution (on $(X,Y)$) by considering two independent standard Gaussians $Z = (Z_1, Z_2)$ and taking $X = \langle{Z,u_1\rangle}$ and $Y = \langle{Z,u_2\rangle}$.
Thus, the probability that $\sgn(X)=\sgn(Y)$ is the same as the probability over a random $Z = (Z_1, Z_2)$ that $\sgn(\langle{Z,u_1\rangle})  = \sgn(\langle{Z,u_2\rangle})$, which is the same if we sample $Z$ according to the uniform distribution on the sphere.
The latter probability is exactly $1-\alpha/\pi$ where $\alpha$ is the angle between $u_1$ and $u_2$.
Thus, the probability is $1-\arccos(\rho)/\pi$, and 
$$
\E[\sgn(X) \cdot \sgn(Y)] = 2\Pr[\sgn(X)=\sgn(Y)]-1 = 1-2\arccos(\rho)/\pi.$$
For $\rho\ge 0$ we have $\E[\sgn(X) \cdot \sgn(Y)]\ge \frac{2}{\pi} \rho$ and for $\rho\le 0$ we get $\E[\sgn(X) \cdot \sgn(Y)]\le \frac{2}{\pi} \rho$. Thus, in both cases,
$\rho \cdot \E[\sgn(X) \cdot \sgn(Y)] \ge \frac{2}{\pi} \rho^2$.
\end{proof}

\subsection{The Rorrelation of Vectors Sampled from $\cU_k$}
We begin with the following simple claim.
\begin{claim}\label{claim:expectation under uniform}Let $\U$ be an $N$-by-$N$ orthogonal matrix. Then,
$$\E_{z\sim \cU_k} \left[\phi_{\U}(z^{(1)},\ldots, z^{(k)})\right] = 
0$$
\end{claim}
\begin{proof}
	\[
	\E_{z\sim \cU_k} [\phi_{\U}(z^{(1)},\ldots, z^{(k)})] = \frac{1}{N}\cdot \sum_{i_1, \ldots, i_k} \E[z^{(1)}_{i_1}] \cdot \U_{i_1,i_2} \cdot \E[z^{(2)}_{i_2}]\cdot  \U_{i_2,i_3}  \cdots \U_{i_{k-1},i_k} \cdot \E[z^{(k)}_{i_k}] = 0.\qedhere
\]
	\end{proof}

Furthermore, we show that for $z^{(1)}, \ldots, z^{(k)}$ drawn from the uniform distribution, the value of $\phi_{\U}(z^{(1)}, \ldots, z^{(k)})$ is concentrated around $0$.
To show that it suffices to bound the variance of $\phi_{\U}(z^{(1)}, \ldots, z^{(k)})$ under the uniform distribution, as we do next.

\begin{claim}\label{claim:concentration under uniform}
Let $\U$ be an $N$-by-$N$ orthogonal matrix. Then,
$$
\var_{z\sim \cU_k} \left[\phi_{\U}(z^{(1)},\ldots, z^{(k)})\right] = 
\E_{z\sim \cU_k} \left[(\phi_{\U}(z^{(1)},\ldots, z^{(k)}))^2\right] =
1/N.
$$
\end{claim}
\begin{proof} Since $\phi_U$ is multilinear we can apply Parseval's identity to get
	\begin{align*}
	\var_{z\sim \cU_k} \left[\phi_{\U}(z^{(1)},\ldots, z^{(k)})\right] &=\E_{z\sim \cU_k} ,\left[\left(\frac{1}{N}\cdot \sum_{i_1, \ldots, i_k} z^{(1)}_{i_1} \cdot \U_{i_1,i_2} \cdot z^{(2)}_{i_2} \cdot  \U_{i_2,i_3}  \cdots \U_{i_{k-1},i_k} \cdot z^{(k)}_{i_k}\right)^2\right]\\&\qquad=\frac{1}{N^2} \sum_{i_1, \ldots, i_k}\U_{i_1,i_2}^2 \U_{i_2,i_3}^2 \cdots \U_{i_{k-1},i_k}^2 = 1/N.\qedhere
	\end{align*}
\end{proof}
Overall, we get that a vector $z$, drawn from the uniform distribution, satisfies $|\phi_{\U}(z^{(1)},\ldots, z^{(k)})| \le 2^{-(k+1)}$  with high probability (at least $1-4^{(k+1)}/N$) .

\section{$\cDUk$ is Pseudorandom for Randomized Decision Trees}
\label{sec:lower bound}
\subsection{Fourier Growth of Decision Trees}

We start by stating two bounds on the Fourier coefficients of decision trees. These bounds capture the fact that the Fourier spectrum of deterministic and randomized decision trees is ``sparse''. More precisely, we bound the sum of absolute values of coefficients of degree $\ell$, and since the sum of squares is at most $1$, this means that within the $\ell$-th level, the Fourier mass is concentrated on a small fraction of the coefficients.

\begin{theorem}[Level-$\ell$ Inequality for Decision Trees -- Version 1]\label{thm:first level-l inequality}
Let $f$ be a (deterministic) decision tree of depth $d$ over $m$ variables $x_1, \ldots, x_m$. Then,
$$\forall \ell \in \{0,1,\ldots, d\}:
\quad L_{1,\ell}(f)=\sum_{S \subseteq[m]:|S|=\ell} {|\hat{f}(S)|} \le \sqrt{O(d)^{\ell} \cdot O(\log m)^{\ell-1}}
$$\end{theorem}
The above inequality is tight for small values of $\ell$. In particular, it gives a $O(\sqrt{d})$ bound on the first level -- a result that was previously obtained by \cite{ODonnellServedio:07,BTW15} and is known to be tight (see Section~\ref{sec:examples} for examples demonstrating its tightness). For higher values of $\ell$ though, the inequality gets sloppier, and for $\ell \ge \Omega(\sqrt{d/\log n})$ a much simpler argument gives  better bounds.

\begin{claim}[Level-$\ell$ Inequality for Decision Trees - Version 2]\label{claim:second level-l inequality}
Let $f$ be a (deterministic) decision tree of depth $d$ over $m$ variables $x_1, \ldots, x_m$. Then,
$$\forall \ell \in \{0,1,\ldots, d\}:
\quad L_{1,\ell}(f)=\sum_{S \subseteq[m]:|S|=\ell} {|\hat{f}(S)|} \le \binom{d}{\ell}
$$
\end{claim}

We defer the proofs of Theorem~\ref{thm:first level-l inequality} and Claim~\ref{claim:second level-l inequality} to Section~\ref{sec:Fourier Growth Decision Trees}. We get the following corollary.

\begin{corollary}
	Let $F$ be a {randomized} decision tree of depth $d$ over $m$ variables $x_1, \ldots, x_m$. Then,
\begin{equation}\label{eq:first level-l inequality}
\forall \ell \in \{0,1,\ldots, d\}:
\quad L_{1,\ell}(F)=\sum_{S \subseteq[m]:|S|=\ell} {|\hat{F}(S)|} \le \sqrt{O(d)^{\ell} \cdot O(\log m)^{\ell-1}}
\end{equation}
and
\begin{equation}\label{eq:second level-l inequality}
\forall \ell \in \{0,1,\ldots, d\}:
\quad L_{1,\ell}(F)=\sum_{S \subseteq[m]:|S|=\ell} {|\hat{F}(S)|} \le \binom{d}{\ell}
\end{equation}
\end{corollary}
\begin{proof}
A randomized decision tree is a convex combination of deterministic decision trees. Since $L_{1,\ell}(\cdot)$ is convex, the bounds follow.
\end{proof}

We conjecture that the right bounds are better for higher levels:
\begin{conjecture}[Conjectured Level-$\ell$ Inequality for Decision Trees]
	Let $f$ be a (deterministic/randomized) decision tree of depth $d$ over $m$ variables $x_1, \ldots, x_m$. Then,
$$\forall \ell \in \{0,1,\ldots, d\}:
\quad L_{1,\ell}(f)=\sum_{S \subseteq[m]:|S|=\ell} {|\hat{f}(S)|} \le \sqrt{{d\choose \ell} \cdot O(\log m)^{\ell-1}}
$$
\end{conjecture}

\subsection{Moment Bounds on $\cDUk$}

\paragraph{Good Orthogonal Matrices.}
We define a pseudorandomness property of orthogonal matrices, from which we will deduce moment bounds on the distribution $\cDUk$.
\begin{definition}[Good Orthogonal Matrices]\label{def:good}
Let $\U$ be an $N$-by-$N$ orthogonal matrix.
	We say that $U$ is {\bf good} if for all $k, \ell \in [N]$, any $k$-by-$\ell$ sub-matrix $W$ of $U$ satisfies $\|W\|\le \sqrt{\frac{100 (k+\ell)\ln N}{N}}$.
\end{definition}
It is not difficult to see that the Hadamard matrix is not good. For example, the Hadamard matrix has a $\sqrt{N}\times \sqrt{N}$ sub-matrix $W$, all whose entries equal $+1/\sqrt{N}$, and thus the norm of $W$ equals $1 \gg \sqrt{\frac{100 (\sqrt{N}+\sqrt{N})\ln N}{N}}$.
On the other hand, we prove that most orthogonal matrices are good.

\begin{lemma}[Most Orthogonal Matrices are Good]\label{lemma:most matrices are good}
Let $U$ be a random orthogonal $N$-by-$N$ matrix. Then, with high probability over the choice of $\U$, $\U$ is good.
\end{lemma}
Furthermore, we show that whenever $U$ is good, we get moment bounds on the corresponding distribution $\cDUk$, defined with respect to $U$.

\begin{lemma}[Moment Bounds for Good Matrices]\label{lemma:good matrices imply moment bounds}
Suppose that $\U$ is a good orthogonal matrix and $\cDUk$ is defined with respect to $\U$.
Then, there exists a universal constant $c$, such that for any $\emptyset \neq S\subseteq [kN]$,	
	$$|\hat{\cDUk}(S)| \le \left(\frac{c \cdot |S| \cdot \log N}{N}\right)^{|S|\cdot (1-1/k)/2},$$
	and for any non-empty set $S$ of size less than $k$, we have $\hat{\cDUk}(S)=0$.
\end{lemma}
We defer the proofs of both lemmata to Section~\ref{sec:moment bounds}.
\subsection{Pseudorandomness of $\cDUk$}
\begin{theorem}\label{thm:main}
Suppose that $\U$ is a good orthogonal matrix and $\cDUk$ is defined with respect to $\U$. 
Let $F$ be a {randomized} decision tree of depth $d$ over $kN$ variables.
Suppose that $d = o(N^{2(k-1)/(3k-1)}/\log(kN))$.
Then, 
$$
\E[F(\cU_k) - F(\cDUk)] \le 
\sqrt\frac{O(d\cdot \log(kN))^{(3k-1)/2}}{N^{k-1}}
$$
\end{theorem}
\begin{proof}
We have
\begin{align*}
|\E[F(\cU_k) - F(\cDUk)]|
&= \big|\sum_{S\neq \emptyset} {\hat{F}(S) \cdot \hat{\cDUk}(S)}\big|\\
&\le \sum_{\ell=k}^{kN} \sum_{|S|=\ell}{|\hat{F}(S) \cdot \hat{\cDUk}(S)|} \\
&\le \sum_{\ell=k}^{kN} \left(\max_{|S|=\ell} |\hat{\cDUk}(S)|\right) \cdot \sum_{|S|=\ell}{|\hat{F}(S)|}\\
&\le \sum_{\ell=k}^{kN}	\left(\frac{c \cdot \ell \cdot \ln N}{N}\right)^{\ell\cdot (1-1/k)/2} \cdot  L_{1,\ell}(F)
\end{align*}
Now we break the right hand side above to two sub-sums:
\begin{enumerate}
	\item for $\ell \le \sqrt{d/\log(kN)}$ we will use the bounds on $L_{1,\ell}(F)$ from Eq.~\eqref{eq:first level-l inequality}.
	\item for $\ell > \sqrt{d/\log(kN)}$ we will use the bounds on $L_{1,\ell}(F)$ from Eq.~\eqref{eq:second level-l inequality}.
\end{enumerate}
That is, we bound the lower order terms by
\begin{align*}
&\sum_{\ell=k}^{\sqrt{d/\log(kN)}}	\left(\frac{c \cdot \ell \cdot \ln N}{N}\right)^{\ell\cdot (1-1/k)/2} \cdot  L_{1,\ell}(F)	\\
&\le \sum_{\ell=k}^{\sqrt{d/\log(kN)}}	\left(\frac{c \cdot \ell \cdot \ln N}{N}\right)^{\ell\cdot (1-1/k)/2} \cdot  O(d \cdot \log(kN))^{\ell/2}\\
&\le \sum_{\ell=k}^{\sqrt{d/\log(kN)}}	\left(\frac{c \cdot \sqrt{d/\log(kN)} \cdot \ln N}{N}\right)^{\ell\cdot (1-1/k)/2} \cdot  O(d \cdot \log(kN) )^{\ell/2}\\
&\le \sum_{\ell=k}^{\sqrt{d/\log(kN)}} O\left(\frac{(d\cdot \log(kN))^{1+(1-1/k)/2}}{N^{1-1/k}}\right)^{\ell/2} \\
&\le O\left(\frac{(d\cdot \log(kN))^{1+(1-1/k)/2}}{N^{1-1/k}}\right)^{k/2} = \sqrt{\frac{O(d\cdot \log(kN))^{(3k-1)/2}}{N^{k-1}}}
\end{align*}
where in the last inequality, the assumption that $d = o(N^{2(k-1)/(3k-1)}/\log(kN))$ is used to deduce that this is a decreasing geometric progression.
We bound the higher order terms by
\begin{align*}
&\sum_{\ell > \sqrt{d/\log(kN)}}	\left(\frac{c \cdot \ell \cdot \ln N}{N}\right)^{\ell\cdot (1-1/k)/2} \cdot  L_{1,\ell}(F)	\\
&\le \sum_{\ell > \sqrt{d/\log(kN)}}	\left(\frac{c \cdot \ell \cdot \ln N}{N}\right)^{\ell\cdot (1-1/k)/2} \cdot  \binom{d}{\ell} \\
&\le \sum_{\ell > \sqrt{d/\log(kN)}}\left(\left(\frac{c \cdot \ell \cdot \ln N}{N}\right)^{(1-1/k)/2} \cdot \frac{e\cdot d}{\ell}\right)^{\ell} \\
&\le \sum_{\ell > \sqrt{d/\log(kN)}}	O\left(\left(\frac{\log N}{N}\right)^{(1-1/k)/2} \cdot \frac{d}{\ell^{1-(1-1/k)/2}}\right)^{\ell} \\
&\le \sum_{\ell > \sqrt{d/\log(kN)}}	O\left(\left(\frac{\log N}{N}\right)^{(1-1/k)/2} \cdot \frac{d}{\sqrt{d/\log(kN)}^{1-(1-1/k)/2}}\right)^{\ell} \\
&\le \sum_{\ell > \sqrt{d/\log(kN)}}	O\left(\frac{(d\cdot \log(kN))^{1+(1-1/k)/2}}{N^{1-1/k}}\right)^{\ell/2} \\
&\le O\left(\frac{(d\cdot \log(kN))^{1+(1-1/k)/2}}{N^{1-1/k}}\right)^{k/2}  
= \sqrt{\frac{O(d\cdot \log(kN))^{(3k-1)/2}}{N^{k-1}}}\;.\qedhere
\end{align*}
\end{proof}
Similarly, in Appendix~\ref{app:B}, we show that assuming Conjecture~\ref{conj},  for any good $\U$, the distribution $\cDUk$ is pseudorandom against any depth $N^{1-1/k}/\polylog(N)$ decision tree.

\subsection{Shallow Randomized Decision Tree Cannot Solve the $k$-fold Rorrelation Problem}

We prove the following lower bound on the depth of randomized decision trees solving the $k$-fold Rorrelation Problem.

\begin{theorem}\label{thm:lb}
Let $\U$ be a good orthogonal $N$-by-$N$ matrix.
Let $k\ge 2$ be such that $16 \cdot 8^k\le N$. 
Suppose that $F$ is a randomized decision tree of depth $d$ solving the $k$-fold Rorrelation problem with success probability at least $2/3$.
Then, $d\ge \Omega(N^{2(k-1)/(3k-1)}/(k\log(kN)))$.
\end{theorem}

Towards proving Theorem~\ref{thm:lb}, we show that $\frac{1}{2}\cU_k + \frac{1}{2}\cDUk$ is a somewhat hard-distribution for any depth-$d$ randomized decision trying to solve the Rorrelation problem, as long as $d = o(N^{2(k-1)/(3k-1)}/\log(kN))$.

\begin{claim}\label{claim:hard-dist} Assume that $16 \cdot 8^k \le N$, and $\U$ is a good orthogonal $N$-by-$N$ matrix.
	Let $F$ be a randomized decision tree of depth 
	$d = o(N^{2(k-1)/(3k-1)}/\log(kN))$.
	Then,
\begin{equation}\label{eq:hard}
\Pr[\text{$z$ is legal input to Rorrelation, and $F$ misclassifies $z$}] \ge \frac{1}{8} \cdot 2^{-k}
\end{equation}
	where the probability is taken over the randomness of $z\sim \frac{1}{2}\cU_k + \frac{1}{2}\cDUk$ and the internal randomness of $F$.
\end{claim}

Before proving Claim~\ref{claim:hard-dist}, we show that it implies Theorem~\ref{thm:lb}.
\begin{proof}[Proof of Thm.~\ref{thm:lb}]
	Suppose that $F$ is a depth $d$ randomized decision tree with success probability at least $2/3$. 
	Then, one can amplify the success probability to at least $1-\frac{1}{10} \cdot 2^{-k}$, by running $F$ sequently $\Theta(k)$ many times and taking the majority vote.
	Thus, we get a randomized decision tree $F'$ of depth $d' = \Theta(d\cdot k)$ such that
	\begin{itemize}
			\item On any YES instance, $F'$ accepts with probability at least $1-\frac{1}{10} \cdot 2^{-k}$.
			\item On any NO instance, $F'$ accepts with probability at most $\frac{1}{10} \cdot 2^{-k}$.
	\end{itemize}
	In particular, Eq.~\eqref{eq:hard} does not hold for $F'$, which by Claim~\ref{claim:hard-dist} implies $d' \ge \Omega( N^{2(k-1)/(3k-1)}/\log(kN))$. Recalling that $d'= \Theta(d\cdot k)$ completes the proof.
\end{proof}

\begin{proof}[Proof of Claim~\ref{claim:hard-dist}]
Assume by contradiction that this is not the case. Then, there exists a randomized decision tree $F$, with
$$\Pr[\text{$z$ is legal input to Rorrelation, and $F$ misclassifies $z$}] < \frac{1}{8}\cdot 2^{-k}.$$
By averaging, there exists a deterministic decision tree $f$ with 
$$
\Pr_{z\sim \frac{1}{2}\cU_k + \frac{1}{2}\cDUk}[\text{$z$ is legal input to Rorrelation, and $f$ misclassifies $z$}] < \frac{1}{8}\cdot 2^{-k}.
$$
In particular, this means that on the uniform distribution 
$$
\Pr_{z\sim \cU_k}[\text{$z$ is legal input to Rorrelation, and $f$ misclassifies $z$}] \le \frac{1}{4}\cdot 2^{-k}
$$
and on the distribution $\cDUk$ we have 
$$
\Pr_{z\sim \cDUk}[\text{$z$ is legal input to Rorrelation, and $f$ misclassifies $z$}] \le \frac{1}{4}\cdot 2^{-k}.
$$
We will show that $f$ distinguishes between $\cDUk$ and $\cU_k$ which will be a contradiction to Theorem~\ref{thm:main}.

For $z\sim \cU_k$ we have that with  probability at least $1-4^{(k+1)}/N$, $|\phi_{\U}(z)| \le  2^{-(k+1)}$. This is a consequence of the concentration inequality we got in Claim~\ref{claim:concentration under uniform} stating that $\E_{z\sim\cU_k} [\phi_{\U}(z)^2] = 1/N$.
Thus, with probability at least $1-4^{(k+1)}/N$ we have that $z$ is a NO instance to the Rorrelation problem, and by the assumption on $f$ with probability at least $1-4^{(k+1)}/N - \frac{1}{4}\cdot 2^{-k}$ it answers NO. That is,
\begin{equation}\label{eq:upper bound acceptance \U_k}
	\E_{z\sim \U_k}[f(z)] \le \frac{4^{k+1}}{N} + \frac{1}{4} \cdot 2^{-k}.
\end{equation}

For $z\sim \cDUk$ we have that $\E_{z\sim \cDUk}[\phi_{\U}(z)] \ge (2/\pi)^{k-1} > 2^{-(k-1)}$.
Since $|\phi_{\U}(z)| \le 1$ for all binary vectors, this means that, for $z\sim \cDUk$, with probability at least $2^{-k}$, we have $\phi_{\U}(z) \ge 2^{-k}$ (as otherwise the expectation would be less than $2^{-(k-1)}$).
Put differently, when sampling from $\cDUk$ with probability at least $2^{-k}$ we get a YES instance for Rorrelation. By that $f$ errs on at most $\frac{1}{4} \cdot 2^{-k}$ of the probability mass of $\cDUk$, it means that 
\begin{equation}\label{eq:lower bound acceptance D_k}
\E_{z\sim\cDUk}[f(z)] \ge 2^{-k} - \frac{1}{4} \cdot 2^{-k}.
\end{equation}
Combining Equations~\eqref{eq:upper bound acceptance \U_k} and \eqref{eq:lower bound acceptance D_k}, we get that 
\begin{equation}\label{eq:distinguishability lower bound}
\E[f(\cDUk)] - \E[f(\cU_k)] \ge \frac{1}{2}\cdot 2^{-k} -  \frac{4^{k+1}}{N} \ge \frac{1}{4}\cdot 2^{-k},
\end{equation}
where in the last inequality we used the assumption $N \ge 16 \cdot 8^k$.
On the other hand, Theorem~\ref{thm:main} shows that 
\begin{equation}\label{eq:distinguishability upper bound}
\E[f(\cDUk)] - \E[f(\cU_k)] \le \sqrt\frac{O(d\cdot \log(kN))^{(3k-1)/2}}{N^{k-1}} \le o(2^{-k}).
\end{equation}
where in the last inequality we used the assumption that $d = o(N^{2(k-1)/3(k-1)}/\log(kN))$.
This yields a contradiction, completing the proof.
\end{proof}

\section{Bounding the Moments of the Distribution $\cDUk$}
\label{sec:moment bounds}

Recall the definition of a good orthogonal matrix from Section~\ref{sec:lower bound}.
In this section, we will prove that most uniform matrices are good (Lemma~\ref{lemma:most matrices are good}), and that for any good uniform matrices, the corresponding distribution $\cDUk$ has bounded moments (Lemma~\ref{lemma:good matrices imply moment bounds}).

We start with Lemma~\ref{lemma:good matrices imply moment bounds}.
Let $U$ be an $N$-by-$N$ orthogonal matrix (as in the definition of $k$-fold Rorrelation).
For $S,T \subseteq [N]$,  we denote by
$\tilde{\U}(S,T) \triangleq \E_{X}\left[\sgn\left(\prod_{i\in S}X_i \prod_{j\in T}(U^{\T}X)_j\right)\right]$ where $X\sim \NN(0,1)^N$ is a standard $N$-dimensional multi-variate Gaussian random vector.

We state our main technical lemma, that connects the property of being good (namely small norms of sub-matrices of $U$), with bounds on $|\tilde{\U}(S,T)|$.
\begin{lemma}\label{lemma:main technical}
Let $S,T \subseteq [N]$.
Then:
\begin{itemize}
\item For any orthogonal matrix $\U$, if $|S|+|T|$ is odd, then $\tilde{\U}(S,T) =0$.
\item For any orthogonal matrix $\U$, if $|S|+|T|$ is even, then 
$$
|\tilde{\U}(S,T)|\le 
(50 \cdot \|\U_{S,T}\|)^{\max\{|S|,|T|\}},
$$
where $\U_{S,T}$ denotes the sub-matrix of $\U$ with rows in $S$ and columns in $T$.
\end{itemize}
\end{lemma}
We defer the proof of Lemma~\ref{lemma:main technical} to the next subsection.
From Lemma~\ref{lemma:main technical}, and the definition of good matrices, the following corollary is immediate.
\begin{corollary}
\label{cor:main}
Suppose that $\U$ is a good orthogonal matrix and $\cDUk$ is defined with respect to $\U$.
Then, there exists a universal constant $c>0$ such that
	\begin{equation}\label{eq:good}
		|\tilde{\U}(S,T)| \le  \left(\frac{c (|S|+|T|)\ln N}{N}\right)^{\max\{|S|,|T|\}/2}.
	\end{equation}
	for all $S, T\subseteq [N]$.
\end{corollary}
Furthermore, from corollary~\ref{cor:main}, we derive bounds on the moments of $\cDUk$ for good matrices.

\begin{lemmaNoNum}
[Lemma~\ref{lemma:good matrices imply moment bounds}, restated]
Suppose that $\U$ is a good orthogonal matrix and $\cDUk$ is defined with respect to $\U$.
Then, there exists a universal constant $c$, such that for any $\emptyset \neq S\subseteq [kN]$,	
	$$|\hat{\cDUk}(S)| \le \left(\frac{c \cdot |S| \cdot \log N}{N}\right)^{|S|\cdot (1-1/k)/2},$$
	and for any non-empty set $S$ of size less than $k$, we have $\hat{\cDUk}(S)=0$.
\end{lemmaNoNum}

\begin{proof}
Recall that $\cDUk$ is a distribution over $k$ blocks of $N$ variables each. We write $S = S_1 \cup S_2 \cup \ldots \cup S_k$ where $S_i$ denotes the intersection of $S$ with the $i$-th block of variables.	
\begin{align*}
\hat{\cDUk}(S) &= 
\E\left[\sgn\left(\prod_{i_1 \in S_1} X^{(1)}_{i_1} \prod_{i_2 \in S_2} X^{(2)}_{i_2} Y^{(1)}_{i_2} \cdots  \prod_{i_{k-1} \in S_{k-1}} X^{(k-1)}_{i_{k-1}} Y^{(k-2)}_{i_{k-2}}\prod_{i_k \in S_k} Y^{(k-1)}_{i_k}\right)\right]
\\ &= \E\left[\sgn\left(\prod_{i_1 \in S_1} X^{(1)}_{i_1} \prod_{i_2 \in S_2} Y^{(1)}_{i_2}\right)\right] \cdots  \;\E\left[\sgn\left(\prod_{i_{k-1} \in S_{k-1}} X^{(k-1)}_{i_{k-1}}\prod_{i_k \in S_k} Y^{(k-1)}_{i_k}\right)\right]\\
&= \tilde{\U}(S_1, S_2) \cdot \tilde{\U}(S_2, S_3) \cdots \tilde{\U}(S_{k-1}, S_k).
\end{align*}
So we see that $\hat{\cDUk}(S)\neq 0$ only if all $|S_i|$ have the same parity. In particular, it equals $0$ for any set $S$ of size between $1$ and $k-1$. In the case where all the $|S_i|$'s have the same parity, 
\begin{align*}
|\hat{\cDUk}(S)|&\le \prod_{i=1}^{k-1} \left(\frac{c \cdot (|S_i|+|S_{i+1}|) \cdot \log N}{N}\right)^{\frac{1}{2} \cdot \max\{|S_i|,|S_{i+1}|\}} \\
&\le \left(\frac{c \cdot |S|  \cdot \log N}{N}\right)^{\frac{1}{2} \cdot(\max\{|S_1|,|S_2|\}+\max\{|S_2|,|S_3|\}+\ldots + \max\{|S_{k-1}|,|S_k|\})}\\
&\le \left(\frac{c \cdot |S|  \cdot \log N}{N}\right)^{\frac{1}{2}  \cdot |S|\cdot (k-1)/k}
\end{align*}
where the last inequality is justified in the following lemma (Lemma~\ref{lemma:all same}).
\end{proof}

\begin{lemma}\label{lemma:all same}
Let $a_1, a_2, \ldots, a_k\in \R$.
Then, $$\max(a_1,a_2)+\max(a_2,a_3) + \ldots + \max(a_{k-1},a_k) \ge (a_1 + a_2 + \ldots + a_k) \cdot \frac{k-1}{k}.$$
\end{lemma}
\begin{proof}
	Let $i\in [k]$ be an index such that $a_i = \min(a_1, \ldots, a_k)$.
	We claim that
	\begin{equation}\label{eq:1}
	\max(a_1,a_2)+\max(a_2,a_3) + \ldots + \max(a_{k-1},a_k) \ge a_1 + a_2 + \ldots + a_{i-1} + a_{i+1} + \ldots + a_k.\end{equation}
	This is true since for every $j\le i-1$ we have $\max(a_j, a_{j+1}) \ge a_j$ and for every $j\in \{i, \ldots, k-1\}$ we have $\max(a_j,a_{j+1}) \ge a_{j+1}$. Combining these inequalities together gives Eq.~\eqref{eq:1}.
	Finally, observe that if $s \triangleq a_1 + \ldots + a_k$, then $a_i \le s/k$ and thus
	$$\max(a_1,a_2)+\max(a_2,a_3) + \ldots + \max(a_{k-1},a_k) \ge s - s/k,$$
	which completes the proof.
\end{proof}

\subsection{Proof of Lemma~\ref{lemma:main technical}}
\begin{proof}
We denote by $Y_i = (U^{\T} X)_i$ for $i \in [N]$.
The multivariate Gaussian distribution $(X_1, \ldots, X_N, Y_1, \ldots Y_N)$ is symmetric around $\vec{0}$.
Thus, when $|S|+|T|$ is odd, the product  $(\prod_{i\in S} X_i \prod_{j\in T} Y_j)$ is an odd function, and hence the expectation of its sign equals $0$.

For the rest of the proof, we assume that $|S|+|T| = 2\ell$ is even.
We assume without loss of generality that $|S|\ge |T|$, and denote by $\ell' = \max\{|S|,|T|\} = |S|$. 
We also assume without loss of generality that $\|U_{S,T}\|\le 1/50$ since otherwise the claim is trivial.
We compute  $\tilde{\U}(S,T)$ as a Lebesgue integral.
For brevity, we denote by $dX_S = \prod_{i\in S} dX_i$ and by $dY_T = \prod_{i\in T} dY_j$. 
We also denote by $X_S$ the vector $(X_i)_{i\in S}$ and  by $Y_T$ the vector $(Y_i)_{j\in T}$ and by $X^S$ and $Y^T$ the products
$\prod_{i \in S} X_i$ and $\prod_{i \in T} Y_j$ respectively.

The covariance matrix of the Gaussian random variables $\{X_i\}_{i \in S}$ and $\{Y_j\}_{j \in T}$ is given by 
$\Sigma = \left(\begin{matrix}
	I_{S} & \U_{S,T}\\
	\U_{S,T}^{\T}& I_T\end{matrix}\right).$
Due to the assumption $\|\U_{S,T}\|\le 1/50$ we get that $\Sigma$ is non-singular. Moreover, since $\|U_{S,T}\|\le 1/50$, all eigenvalues of $\Sigma$ are in the range $[0.98,1.02]$, and we get that the determinant of $\Sigma$, denoted $|\Sigma|$, is in the interval $[0.98^{2\ell}, 1.02^{2\ell}]$.

Next, we write $\Sigma^{-1}$ in terms of $U_{S,T}$.
By the matrix inversion formula of block matrices
$$
\Sigma^{-1} = 
\left(\begin{matrix}
(I_{S} -\U_{S,T}\U_{S,T}^{\T})^{-1} 
& -\U_{S,T} \cdot (I_T - \U_{S,T}^{\T}\U_{S,T})^{-1}
\\
-\U_{S,T}^{\T} \cdot (I_{S} -\U_{S,T}\U_{S,T}^{\T})^{-1}
& (I_T - \U_{S,T}^{\T}\U_{S,T})^{-1}
\end{matrix}
\right)
\triangleq 
\left(\begin{matrix}
A & C	\\
C^{\T} & B
\end{matrix}
\right).
$$
By Woodbury matrix identity 
\begin{align}\label{eq:A'}
A = I_S+U_{S,T}\left(I_T-U_{S,T}^{\T} U_{S,T}\right)^{-1}U_{S,T}^{\T} = I_S+U_{S,T}\cdot B\cdot  U_{S,T}^{\T} \triangleq I_S + A'	\\
\label{eq:B'}
B = I_T+U_{S,T}^{\T} \left(I_S-U_{S,T} U_{S,T}^{\T}\right)^{-1}U_{S,T} = I_T+U_{S,T}^{\T}\cdot A\cdot  U_{S,T} \triangleq I_T + B'
\end{align}
where $A,B,A',B'$ are symmetric PSD matrices.
Furthermore, observe that \begin{equation}\label{eq:C}C = - U_{S,T} B\end{equation} and that $\|B\| \le 1.01$. We are ready to start analyzing $\tilde{\U}(S,T)$ as a Lebesgue integral:
\begin{align*}
&\left|\tilde{\U}(S,T)\right| =\left|\E_{X,Y}
\left[\sgn\left(X^S Y^T\right)\right]\right|	 \\
&=  \frac{1}{\sqrt{|\Sigma|\cdot (2\pi)^{2\ell}}}\cdot \left|
\int_{(-\infty,\infty)^{2\ell}}
{\sgn(X^S Y^T) \cdot 
e^{\frac{-\|X_S\|_2^2 -\|Y_T\|_2^2}{2}} 
\cdot 
e^{-(X_S^{\T} A' X_S/2 
\;+\;Y_T^{\T} B' Y_T/2 
\;+\;X_S^{\T} C Y_T)}
\;dX_S \;dY_T
}\right|\end{align*}
Now,
\begin{align}
&\left|\int_{(-\infty,\infty)^{2\ell}}
{\sgn(X^S Y^T) \cdot 
e^{\frac{-\|X_S\|_2^2 -\|Y_T\|_2^2}{2}} 
\cdot 
e^{-(X_S^{\T} A' X_S/2 
\;+\;Y_T^{\T} B' Y_T/2 
\;+\;X_S^{\T} C Y_T)}
\;dX_S \;dY_T
}\right|
\nonumber\\
&= 
\left|\sum_{\substack{a \in \{-1,1\}^{S},\\ b\in \{-1,1\}^{T}}}a^{S}\cdot b^{T} \int\limits_{[0,\infty)^{2\ell}}{ 
e^{\frac{-\|X_S\|_2^2 -\|Y_T\|_2^2}{2}} 
\cdot 
e^{-((a.X)_S^{\T} A' (a.X)_S/2 
\;+\;(b.Y)_T^{\T} B' (b.Y)_T/2 
\;+\;(a.X)_S^{\T} C (b.Y)_T)}
\;dX_S \;dY_T
}\right|\nonumber\\
&\le  \int\limits_{[0,\infty)^{2\ell}}{ e^{\frac{-\|X_S\|_2^2 - \|Y_T\|_2^2}{2}}
\cdot \left|
\sum_{
\substack{a \in \{-1,1\}^{S},\\ b\in \{-1,1\}^{T}}} a^S \cdot b^T\cdot 
e^{-((a.X)_S^{\T} A' (a.X)_S/2 
\;+\;(b.Y)_T^{\T} B' (b.Y)_T/2 
\;+\;(a.X)_S^{\T} C (b.Y)_T)}
\right|
dX_S \;dY_T
}\nonumber\\
&\le  \int\limits_{[0,\infty)^{2\ell}}{
e^{\frac{-\|X_S\|_2^2 - \|Y_T\|_2^2}{2}}
\sum_{b\in \{-1,1\}^{T}}  e^{\frac{-(b.Y)_T^{\T} B' (b.Y)_T}{2}} \cdot 
\left|\sum_{a\in \{-1,1\}^{S}}
a^S\cdot
e^{-((a.X)_S^{\T} A' (a.X)_S/2 
\;+\;(a.X)_S^{\T} C (b.Y)_T)}
\right|
dX_S \;dY_T
}\nonumber
\\
&\le
\int\limits_{[0,\infty)^{2\ell}}{
e^{\frac{-\|X_S\|_2^2 - \|Y_T\|_2^2}{2}}
\sum_{b\in \{-1,1\}^{T}} 
\left|\sum_{a\in \{-1,1\}^{S}}
a^S\cdot
e^{-((a.X)_S^{\T} A' (a.X)_S/2 
\;+\;(a.X)_S^{\T} C (b.Y)_T)}
\right|
dX_S \;dY_T} \label{eq:integral}
\end{align}
where in the last inequality we used the fact that $B'$ is PSD.
Fix $b \in \{-1,1\}^{T}$, $Y_T$ and $X_S$. 
We analyze the internal sum 
$$
\sum_{a \in \{-1,1\}^{S}}a^S
\cdot e^{-((a.X)_S^{\T} A' (a.X)_S/2 
\;+\;(a.X)_S^{\T} C (b.Y)_T)}
$$
using Taylor expansion. 	
For each $a\in\{-1,1\}^S$, we develop the Taylor series of  
$$
e^{Q(a) + L(a)}\quad\text{where}\qquad
Q(a)=-((a.X)_S)^{\T} A' (a.X)_S/2,  \qquad 
L(a)=-(a.X)_S^{\T} C (b.Y)_T
$$
and obtain 
$$
\sum_{a \in \{-1,1\}^{S}}a^S
\cdot e^{Q(a) +L(a)} = 
\sum_{a \in \{-1,1\}^{S}}a^S
\cdot \sum_{i=0}^{\infty} \frac{(Q(a)+L(a))^i}{i!} =
\sum_{a \in \{-1,1\}^{S}}a^S
\cdot \sum_{0 \le i_1,i_2} \binom{i_1+i_2}{i_1}  \frac{Q(a)^{i_1} L(a)^{i_2}}{(i_1+i_2)!}
$$
Observe that $Q(a)^{i_1} \cdot L(a)^{i_2}$ is a polynomial of degree $2i_1+i_2$ in the variables $a$. Summing over all $a\in \{-1,1\}^{S}$,  the terms corresponding to $(i_1,i_2)$ with $2i_1+i_2 <\ell'$ cancel out, and  we are left  only with the terms corresponding to $i_1, i_2$ such that $2i_1+i_2\ge \ell'$. Thus, we get 	
\begin{align*}
\left|\sum_{a \in \{-1,1\}^{S}}
a^S \cdot
e^{-((a.X)_S^{\T} A' (a.X)_S/2 
\;+\;(a.X)_S^{\T} C (b.Y)_T)}\right| &\le 
\sum_{a \in \{-1,1\}^{S}}\sum_{\substack{i_1,i_2:\\2i_1+i_2\ge \ell'}}
\binom{i_1+i_2}{i_1} \frac{|Q(a)^{i_1}\cdot L(a)^{i_2}|}{(i_1+i_2)!} 
\end{align*}
Plugging this bound in Eq.~\eqref{eq:integral} gives
\begin{align*}
&|\tilde{\U}(S,T)| \le 	
\frac{1}{\sqrt{|\Sigma|(2\pi)^{2\ell}}} \int\limits_{[0,\infty)^{2\ell}}{ e^{\frac{-\|X_S\|_2^2 - \|Y_T\|_2^2}{2}} 
\cdot \sum_{a,b} \sum_{\substack{i_1,i_2:\\2i_1+i_2\ge \ell'}} 
\binom{i_1+i_2}{i_1}  \frac{|Q(a)^{i_1}\cdot L(a)^{i_2}|}{(i_1+i_2)!} 
\;dX_S \;dY_T}
\\
&=\frac{1}{\sqrt{|\Sigma|(2\pi)^{2\ell}}} \int\limits_{(-\infty,\infty)^{2\ell}}{e^{\frac{-\|X_S\|_2^2 - \|Y_T\|_2^2}{2}} \cdot 
\sum_{\substack{i_1,i_2:\\2i_1+i_2\ge \ell'}} 
\binom{i_1+i_2}{i_1}
\cdot
\frac{|(X_S^{\T} A' X_S/2)^{i_1}\cdot (X_S^{\T} C Y_T)^{i_2}|}{(i_1+i_2)!}
\;dX_S \;dY_T}
\\
&= \frac{1}{\sqrt{|\Sigma|}} \sum_{\substack{i_1,i_2:\\2i_1+i_2\ge \ell'}}
\binom{i_1+i_2}{i_1} 
\int\limits_{(-\infty,\infty)^{2\ell}}{
\frac{1}{\sqrt{(2\pi)^{2\ell}}} \cdot e^{\frac{-\|X_S\|_2^2 - \|Y_T\|_2^2}{2}} 
\cdot \frac{|(X_S^{\T} A' X_S)^{i_1} \cdot (X_S^{\T} C Y_T)^{i_2}|}{2^{i_1}\cdot (i_1+i_2)!}
\; dX_S\; dY_T}\end{align*}
where in the last equality we used Fubini's theorem.
Observe that in the right hand side, each internal integral is with respect to $2\ell$ {\em independent} standard Gaussians.
To avoid confusion, we denote these $2\ell$ standard Gaussians by $\{\tilde{X}_i\}_{i\in S}$ and $\{\tilde{Y}_j\}_{j\in T}$.
We get that
\begin{align*}
|\tilde{U}_{S,T}|&\le \frac{1}{\sqrt{|\Sigma|}}  \sum_{i_1,i_2: 2i_1+i_2\ge \ell'}
\binom{i_1+i_2}{i_1}
\cdot \E_{\tilde{X},\tilde{Y}}\left[ \frac{|(\tilde{X}_S^{\T} A' \tilde{X}_S)^{i_1}\cdot (\tilde{X}_S^{\T} C \tilde{Y}_T)^{i_2}|}{2^{i_1}\cdot (i_1+i_2)!}
\right]\\
&\le
\frac{1}{\sqrt{|\Sigma|}}
\sum_{i_1,i_2: 2i_1+i_2\ge \ell'}\binom{i_1+i_2}{i_1} \cdot 
\frac{\|A'\|^{i_1} \cdot \|C\|^{i_2} \cdot \E\left[\|\tilde{X}_S\|^{2i_1+i_2} \|\tilde{Y}_T\|^{i_2}\right]}{2^{i_1} \cdot (i_1+i_2)!} \\
&\le
\frac{1}{\sqrt{|\Sigma|}}
\sum_{i_1,i_2: 2i_1+i_2\ge \ell'}\binom{i_1+i_2}{i_1} \cdot
\frac{\|A'\|^{i_1} \cdot \|C\|^{i_2} \cdot \sqrt{\E[\|\tilde{X}_S\|^{4i_1+2i_2}] \E[\|\tilde{Y}_T\|^{2i_2}]}}{2^{i_1}\cdot (i_1+i_2)!}
\end{align*}
Now, $\|\tilde{X}_S\|_2^{2}$ is  $\chi$-squared random variable with $|S|$ degrees of freedom. 
It is known that its $2i_1+i_2$ moment equals $|S|\cdot (|S|+2) \cdots (|S|+4i_1+2i_2-2) \le (|S|+2i_1+i_2)^{2i_1+i_2} \le (4i_1+2i_2)^{2i_1 + i_2}$.
Similarly,  $\|\tilde{Y}_T\|_2^{2}$ is  $\chi$-squared random variable with $|T|$ degrees of freedom. It is known that its $i_2$ moment equals 
$
|T|\cdot (|T|+2) \cdots (|T|+2i_2-2) \le (|T|+i_2)^{i_2}
\le (4i_1+2i_2)^{i_2}
$.
Overall, we get
\[|\tilde{\U}(S,T)| \le 
\frac{1}{\sqrt{|\Sigma|}}
\sum_{i_1,i_2:2i_1+i_2\ge \ell'} 
\binom{i_1+i_2}{i_1} \cdot
\frac{\|A'\|^{i_1} \cdot \|C\|^{i_2}}{2^{i_1}}\cdot  \left(\frac{4i_1+2i_2}{(i_1+i_2)/e}\right)^{i_1+i_2}
.\]
Recall that  $\|A'\| \le \|U_{S,T}\|^{2} \cdot \|B\|$
and $\|C\|\le \|U_{S,T}\|\cdot \|B\|$ (by Equations~\eqref{eq:A'} and \eqref{eq:C}) where $\|B\|\le 1.01$.
This gives 
\begin{align*}
|\tilde{\U}(S,T)|
 &\le 0.98^{-\ell} \cdot \sum_{i_1,i_2: 2i_1+i_2\ge \ell'}\binom{i_1+i_2}{i_1} \cdot \frac{\|U_{S,T}\|^{2i_1+i_2} \cdot (1.01)^{i_1+i_2} }{2^{i_1}} \cdot (4e)^{i_1+i_2} \\
&\le 1.03^{\ell} \cdot\sum_{i_1,i_2: 2i_1+i_2\ge \ell'}  \binom{2i_1+i_2}{i_1} \cdot \|U_{S,T}\|^{2i_1+i_2} \cdot \frac{(4.04e)^{2i_1+i_2}}{(8.08e)^{i_1}} \\
&= 1.03^{\ell}\cdot \sum_{d=\ell'}^{\infty} (4.04e)^d \cdot \|U_{S,T}\|^{d}  \sum_{i_1\le d/2} \binom{d}{i_1} \cdot \frac{1}{(8.08e)^{i_1}} 
\\&\le 1.03^{\ell} \cdot\sum_{d=\ell'}^{\infty} (4.04e)^d \cdot \|U_{S,T}\|^{d} \cdot (1+\tfrac{1}{(8.08e)})^{d} \\
&\le 1.03^{\ell}\cdot 2\cdot (12 \|U_{S,T}\|)^{\ell'} \le (50\|U_{S,T}\|)^{\ell'}.\qedhere
\end{align*}
\end{proof}

\subsection{Most Orthogonal Matrices are Good}

We obtain Lemma~\ref{lemma:most matrices are good} by applying a union bound over all possible sub-matrices of a random orthogonal matrix $\U$.
We start by showing that for a fixed subset of rows $S\subseteq [N]$ and a fixed set of columns $T\subseteq [N]$, the norm of $\U_{S,T}$ has sub-Gaussian tails.
\begin{claim}\label{claim:fixed S,T}
Let $\U$ be a random orthogonal $N$-by-$N$ matrix. Let $S, T \subseteq [N]$ be some fixed sets. Then, for all $t>0$,
\[\Pr_U[\|\U_{S,T}\|\ge 2t/\sqrt{N}]\le 2\cdot 9^{|S|+|T|} \cdot e^{-t^2/8}.\]
\end{claim}
\begin{proof}
Let $\eps = 1/4$. 
For brevity, denote by $\U' = \U_{S,T}$.
We take an $\eps$-net $\cal X$ over the unit sphere $\mathcal{S}^{|S|-1}$ and an $\eps$-net $\cal Y$ over the unit sphere $\mathcal{S}^{|T|-1}$.
A simple volume argument shows that there exist such $\eps$-nets with size at most $(1+2/\eps)^{|S|}$ and $(1+2/\eps)^{|T|}$ respectively 
(cf. \cite[Lemma~5.2]{Ver10}).
Another simple argument shows that 
$$\max_{x\in \mathcal{X}, y\in \mathcal{Y}} (x^{\T}   \U'   y) \ge (1-2\eps)\|\U'\| = \|\U'\|/2.$$
To see this note that $\|\U'\|$ is the maximal $(x')^{\T} \U' y'$ taken over all unit vectors $x'\in \R^{S}$ and $y'\in \R^T$.
Let $x'$ and $y'$ be vectors that attain this maximum.
Let $x$ be the closest vector to $x'$ in $\mathcal{X}$.
Let $y$ be the closest vector to $y'$ in $\mathcal{Y}$.
By the definition of $\eps$-nets, $\|x-x'\|\le \eps$ and $\|y-y'\|\le \eps$.
Thus, 
$$
x^{\T} U'y = (x')^{\T} U' y' + (x-x')^{\T} \U' y' + x^{\T} U'(y-y') \ge \|\U'\| - \|\U'\|\cdot \eps - \|\U'\|\cdot \eps\;.
$$

Now, for any fixed unit vectors $x \in \mathcal X$ and $y\in \mathcal Y$, the next claim shows that $\Pr[x^{\T} U y \ge t/\sqrt{N}] \le 2\cdot e^{-t^2/8}$. Note that $x^{\T} U y$ is the same as $x^{\T} U' y$.
Hence, by a simple union bound and Claim~\ref{claim:U11}, we have
$$\Pr[\exists x \in \mathcal{X}, y\in \mathcal{Y}: x^{\T} U'y \ge t/\sqrt{N}]\le |\mathcal{X}||\mathcal{Y}| \cdot 2 \cdot e^{-t^2/8} = 9^{|S|} \cdot 9^{|T|} \cdot 2\cdot  e^{-t^2/8}.$$
Overall, we get that \[\Pr[\|\U'\|\ge 2t/\sqrt{n}]\le  \Pr[\exists x \in \mathcal{X}, y\in \mathcal{Y}: x^{\T} U'y \ge t/\sqrt{N}] \le 9^{|S|} \cdot 9^{|T|} \cdot 2\cdot e^{-t^2/8}.\qedhere\]
\end{proof}

\begin{claim}\label{claim:U11}
	Let $\U$ be a random orthogonal $N$-by-$N$ matrix. Let $x\in \R^N$ and $y\in \R^N$ be any two fixed unit vectors. Then, for all $t>0$,
\[\Pr_U[x^{\T} \U y \ge t/\sqrt{N}]\le 2\cdot e^{-t^2/8}.\]
\end{claim}
\begin{proof}
The main observation is that for a fixed vectors $x$ and  $y$, the distribution of $x^{\T} \U y$ is the same as the distribution of the first coordinate of a random vector on the $(N-1)$-dimensional sphere.
To see it, first note that $Z = \U y$ is a uniform vector on the $(N-1)$-dimensional sphere. 
Furthermore, the inner product of $Z$ with any fixed unit vector $x$ is the same no matter which vector $x$ is chosen. 
In particular, it is the same as the inner product with $e_1$, or in other words, it is distributed the same as $Z_1$.

To finish the proof we show that if $Z$ is a uniformly random vector	 on the $(N-1)$-dimensional sphere, then $$\Pr[Z_{1} \ge t/\sqrt{N}] \le 2e^{-t^2/8}.$$
Note that $Z$ can be sampled by taking $N$ independent Gaussian variables $(Z'_1, \ldots, Z'_N)$ with mean $0$ and variance $1$ and normalizing them.
	Thus $Z_1$ is distributed as $Z'_1/\sqrt{\sum_{i}(Z'_i)^2}$.
	Thus
	\begin{align*}
\Pr[Z_1 \ge t/\sqrt{N}] &\le \Pr\left[|Z'_1| \ge t/2 \;\; \vee \;\; \sum_{i=1}^N (Z'_i)^2 \le N/4\right] \\&\le  \Pr\left[|Z'_1|\ge t/2\right] + \Pr\left[\sum_{i=1}^N (Z'_i)^2 \le N/4\right].\end{align*}
The first summand $\Pr[|Z'_1|\ge t/2]$ is bounded by $e^{-t^2/8}$.
As for the second summand, the random variable $\left(\sum_{i=1}^N{(Z'_i)^2}\right)$ is distributed according to a $\chi^2$-distribution with parameter $N$, and 
	Chernoff bounds on such distributions show that $\Pr[\sum_{i=1}^N (Z'_i)^2 \le N/4] \le (e^{1/4}/4)^{N/2}\le e^{-N/2}$. 
	Overall, we get $\Pr[Z_1 \ge t/\sqrt{N}]\le e^{-t^2/8} + e^{-N/2}$.
	Now, if $t > \sqrt{N}$ then the probability $\Pr[Z_1 \ge t/\sqrt{N}]$ equals $0$, since a vector of norm $1$ cannot have a coordinate with value larger than $1$. So we only need to consider the case where $t \le \sqrt{N}$, in which case $e^{-N/2} \le e^{-t^2/2}$, which completes the proof.
\end{proof}

We are ready to prove Lemma~\ref{lemma:most matrices are good}, which we restate next.
\begin{lemmaNoNum}[Lemma~\ref{lemma:most matrices are good}, restated]
	With high probability over the choice of a random orthogonal $N$-by-$N$  matrix $\U$, we have
$$	\|\U_{S,T}\| < \left(\frac{100 (|S|+|T|)\ln N}{N}\right)^{1/2}$$
	for all non-empty $S, T\subseteq [N]$.
	In other words, with high probability $U$ is good (as in Def.~\ref{def:good}).
\end{lemmaNoNum}
\begin{proof}
For any specific $S,T\subseteq [N]$, by Claim~\ref{claim:fixed S,T}, the probability	that 
$$\|\U_{S,T}\| \ge \left(\frac{100 (|S|+|T|)\ln N}{N}\right)^{1/2}$$
is at most $2\cdot 9^{|S|+|T|} \cdot e^{-(25 (|S|+|T|)\ln N)/8} \le N^{-2(|S|+|T|)}$. This allows a union bound over all non-empty sets as $\sum_{\emptyset \neq S\subseteq [N]} \sum_{\emptyset \neq T\subseteq [N]} N^{-2(|S|+|T|)}=O(1/N^2)$.
\end{proof}


\section{Fourier Coefficients of Decision Trees}\label{sec:Fourier Growth Decision Trees}

In this section, we treat Boolean function as functions mapping $\{-1,1\}^n$ to $\{0,1\}$, 
which will be a lot more convenient in the proofs.%
\footnote{Note that we can transform any function $f:\{-1,1\}^n \to \{-1,1\}$ to a function $f':\{-1,1\}^n \to \{0,1\}$ by taking $f'(x) = \frac{1-f(x)}{2}$. Thus, any bounds on $\sum_{S:|S|=\l} |\hat{f'}(S)|$ can be translated to bounds on $\sum_{S:|S|=\l} |\hat{f}(S)|$ with a multiplicative factor of $2$.}
We will prove the bounds on $\sum_{S:|S|=\ell} |\hat{f}(S)|$ by induction on $\l$. We start with a bound on the first level, i.e., a bound on $\sum_{S:|S|=1} {|\hat{f}(S)|}$, that was previously given in~\cite{ODonnellServedio:07,BTW15}. Our proof, however,  gets a tight dependency on the acceptance probability of the function (improving upon~\cite{BTW15}) that will later play a crucial role in the induction.

\begin{theorem}[Level-1 Inequality for Decision Trees]\label{thm:base}
	Let $T$ be a decision tree of depth $d$ computing a Boolean function $f:\{-1,1\}^n \to \{0,1\}$  with $p=\Pr[f(x)=1]$. Then,
	$$
	\sum_{i=1}^n |\hat{f}(\{i\})| \le O(\sqrt{d} \cdot p \cdot \sqrt{\ln(e/p)})\;.
	$$
\end{theorem}
\begin{proof}
	Since we may negate input bits, without loss of generality all Fourier coefficients $\hat{f}(i)$ are non-negative and it suffices to bound the quantity $\sum_{i} \hat{f}(i)$.
	We may also assume without loss of generality that $T$ is a full binary tree. We can ensure this by querying additional variables in case we have a leaf at depth smaller than $d$.
	
For every leaf $\lambda$ in $T$, we associate a vector $v_{\lambda} \in \{-1,0,1\}^n$ such that $(v_{\lambda})_i = 1$ if the $i$-th bit was queried on the path to $v$ and equaled $1$,  $(v_{\lambda})_i = -1$ if the $i$-th bit was queried on the path to $v$ and equaled $-1$,   and $(v_{\lambda})_i = 0$  otherwise.

Denote by $s_\lambda = \sum_{i}(v_{\lambda})_i$ the sum of fixed variables in leaf $\lambda$.
The crucial observation is that if $\lambda$ is a random leaf, then $s_\lambda$ is distributed as the sum of $d$ independent uniformly random $\{-1,1\}$ variables. (This can be easily verified by induction on $d$.)
Thus, by Chernoff's bound we have $\Pr[s_{\lambda} \ge t] = e^{-t^2/2d}$. Now, note that 
$$\sum_{i=1}^{n} \hat{f}(\{i\}) = \sum_{i=1}^{n} \E_{x}[f(x)\cdot x_i] = \E_{x}\left[f(x)\cdot \sum_{i=1}^{n} x_i\right]  = \E_{\lambda}[f(\lambda) \cdot s_{\lambda}],$$
for a random leaf $\lambda$.
To finish the proof, we will use the tail bounds on $s_\lambda$ and the fact that $\Pr[f(\lambda)=1]=p$ to get a bound on $\E_{\lambda}[f(\lambda) \cdot s_{\lambda}]$.
Set $t = \sqrt{2 d \ln(e/p)}$.
We can rewrite $$
\E_\lambda[f(\lambda) \cdot s_\lambda] = \E[f(\lambda)\cdot t] + \E_{\lambda}[f(\lambda) \cdot (s_\lambda-t)] \le p\cdot t + \E_{\lambda}[\max(0,s_\lambda-t)],$$
and note that 
$$\E_{\lambda}[\max(0,s_\lambda-t)] \le \sum_{x=\lceil{t\rceil}}^{\infty} \Pr[s_{\lambda} \ge x]  \le \sum_{x=\lceil{t\rceil}}^{\infty} e^{-x^2/2d} \le O(p d/\lceil{t\rceil}) \le O(pt)\;.$$
Overall, we get 
$\sum_{i=1}^{n} \hat{f}(i) \le O(pt)$, which completes the proof as we assumed w.l.o.g. that all $\hat{f}(i)$ are non-negative.
\end{proof}

The rest of the proof will build on this base case by decomposing a decision tree into smaller pieces. The proof strategy follows the one laid in \cite{CHRT18} for constant-width read-once oblivious branching programs (ROBPs, in short), in the following aspects:
(i) it is carried by induction on $\l$, (ii) it relies on decomposition results,  (iii) we get bounds with respect to the probability of acceptance, which allows for the induction to be effective, and (iv) it uses the relabeling technique. 
Despite those similarities, there are key differences between the models of constant-width ROBPs and decision trees. 
For example, in the former  the total number of nodes is linear in $n$, whereas in the latter we think of the number of nodes as exponential in the depth $d$ (which we think of as larger than $\sqrt{n}$). 
It is thus surprising that one can import the techniques to our setting. 
In addition, while in ROBPs the variables are read in an oblivious order, in the case of decision tree the order is selected adaptively and may change between different paths of the same tree.
One technique that we could not manage to transfer from the ROBPs case is the ``bootstrapping'' lemma of Reingold, Steinke, and Vadhan \cite{ReingoldSV13}, that showed that bounds on the first $\Theta(\log n)$ levels implies similar bounds on all levels. We believe that this is the main reason we could not achieve the conjectured bounds for high levels, as our bounds for lower levels are nearly tight.

\paragraph{Notation:} For every vertex $v$ in the decision tree, we denote by $B_v(x)$ the indicator that the path determined by $x$  passes through $v$.
We denote by $A_v(x)$ the function computed by the subtree rooted at $v$.%
	\footnote{We use $B_v$ and $A _v$ as shorthand to ``before $v$'' and 	``after $v$''.}
	We denote by $\next(v)$ the index of the variable being read in the node $v$. In the case where $v$ is a leaf, we write $\next(v) = \bot$.
			
	\begin{lemma}[Decomposition Lemma]\label{lemma:decomposition}
	Let $S \subseteq[n]$ be a non-empty subset.
Then,	\begin{equation}\label{eq:decomposition Fourier}
			\hat{f}(S) = \sum_{j\in S} \sum_{v:\next(v)=j} \hat{B_v}(S\setminus \{j\}) \cdot \hat{A_v}(\{j\})\;.
	\end{equation}
	\end{lemma}
	\begin{proof}
	For each $j\in [n]$, we denote by $V_j$ the set of vertices $v$ for which $\next(v)=j$ and all other variables in $\{x_i : i\in S\setminus\{j\}\}$ have been queried on the path reaching $v$.
	Observe that the sets $\{V_j\}_{j\in S}$ are pairwise disjoint and that the set of nodes $\bigcup_{j\in S} V_j$ forms a {\bf frontier} in the sense that no vertex in this set is a descendant of another.
	Denote by $L_{\bot}$ the set of leaves on which we have not queried all the variables in $S$. 
	Observe that any path in the tree must reach exactly one node in either $\bigcup_{j\in S} V_j$ or $L_{\bot}$.
	Thus, we may write 
	$$
	f(x) = \sum_{\lambda \in L_{\bot}} B_{\lambda}(x) \cdot f(\lambda)  + \sum_{j\in S}\sum_{v\in V_j} B_{v}(x) \cdot A_v(x)
	$$
	Taking the inner product of both sides with $\prod_{i\in S}x_i$, we get
	$$
	\hat{f}(S) = \sum_{\lambda \in L_{\bot}} \hat{B_{\lambda}}(S) \cdot f(\lambda)  + \sum_{j\in S}\sum_{v\in V_j} \hat{B_{v}}(S\setminus \{j\}) \cdot  \hat{A_v}(\{j\}).
	$$
	Since $\hat{B_{\lambda}}(S) = 0$ for $\lambda\in L_{\bot}$, the first sum can be omitted and we get
	$$
		\hat{f}(S) = \sum_{j\in S}\sum_{v\in V_j} \hat{B_{v}}(S\setminus \{j\}) \cdot  \hat{A_v}(\{j\}).
		$$
Fix $j$. Observe that for any $v$ such that $\next(v)=j$, either $v \in V_j$, or $B_v$ does not depend on some $x_i$ for $i\in S\setminus \{j\}$.
In the latter case $\hat{B_{v}}(S\setminus \{j\}) = 0$, which means that 
			$$\sum_{v: \next(v)=j} \hat{B_{v}}(S\setminus \{j\}) \cdot  \hat{A_v}(\{j\}) = \sum_{v\in V_j} \hat{B_{v}}(S\setminus \{j\}) \cdot  \hat{A_v}(\{j\})$$
			and proves the validity of Eq.\eqref{eq:decomposition Fourier}.\end{proof}

For any vertex $v$ in the tree, we denote by $p_v$ the probability of reaching $v$ under a uniformly chosen input. Alternatively, $p_v = \Pr_{x\in \{-1,1\}^n}[B_v(x)=1] = \hat{B_v}(\emptyset)$.
As a special case of Eq.~\eqref{eq:decomposition Fourier} we get 
\begin{equation}
\hat{f}(\{j\}) = \sum_{v:\next(v)=j} p_v \cdot \hat{A_v}(\{j\}).
\label{eq:Fourier Decomposition singleton}
\end{equation}
Thus, a bound on $\sum_{j} \hat{f}(\{j\})$ implies a bound on 
$\sum_{j=1}^{n}\sum_{v:\next(v)=j} p_v \cdot \hat{A_v}(\{j\}).$
However, in the sequel we will need a stronger version of the bound, where $\hat{A_v}(\{j\})$ are taken with absolute values.
We obtain such a bound as a corollary of Theorem~\ref{thm:base}.
\begin{corollary}[Refinement of Level-1 Inequality for Decision Trees]\label{cor:5.2}
$$
\sum_{j} \sum_{v: \next(v)=j} p_v \cdot |\hat{A_v}(j)| \le O( \sqrt{d} \cdot p \cdot \sqrt{\ln(e/p)})
$$
\end{corollary}
\begin{proof}
Given a decision tree $T$ of depth $d$ for $f$, we generate a new decision $T'$ of depth $d$ that computes another Boolean function $f':\{-1,1\}^n\to \{0,1\}$.  
We obtain $T'$ from $T$ by relabeling some of the edges.
For every non-leaf $v$, with $j = \next(v)$, if $\hat{A_v}(\{j\}) <0$, then we relabel the two edges going from $v$. That is, we swap the labels on the edges marked by $x_j=-1$ and $x_j=1$.
This procedure generates a decision tree $T'$, with the same graph structure, but with different edge labels, and in particular it would most likely compute a different function $f'\neq f$.
Nonetheless, the relabeling does not change the probability of reaching any vertex $v$, $p_v$, or the probability of acceptance. 
Denote by $A'_v$ and $B'_v$ the indicators associated with $v$ in the tree $T'$.
We have that $\Pr[B'_v(x)=1] = \Pr[B_v(x)=1] = p_v$. In addition, if $\next(v)=j$, then $\hat{A'_v}(\{j\}) = |\hat{A_v}(\{j\})|$. 
Thus, all the coefficients $\hat{A'_v}(\{j\})$ are positive and we  conclude that \begin{align*}
\sum_{j=1}^{n} \sum_{v:\next(v)=j} p_v\cdot  |\hat{A_v}(\{j\})| &= 
\sum_{j=1}^{n}\sum_{v:\next(v)=j}  p_v \cdot \hat{A'_v}(\{j\}) \\
&= \sum_{j=1}^{n} \hat{f'}(\{j\}) \tag{Eq.~\eqref{eq:Fourier Decomposition singleton} on $f'$}\\
&\le O( \sqrt{d} \cdot p \cdot \sqrt{\ln(e/p)}) 
\tag{Thm.~\ref{thm:base} on $f'$}.
\end{align*}
\end{proof}

\begin{corollary}[Further Refinement of Level-1 Inequality for Decision Trees]\label{cor:5.3}
Let $T$ be a decision tree of depth at most $d$ computing a Boolean function $f$.
Let $V_0,V_1, \ldots, V_{d}$ denote the set of vertices of depth $0,1, \ldots, d$ in $T$ respectively.
Let $0 \le d' <  d'' \le d$.
Then,
$$
\sum_{v \in (V_{d'} \cup V_{d'+1} \cup \ldots \cup V_{d''-1}) } p_v \cdot |\hat{A_v}(\{\next(v)\})| \le O( \sqrt{d''-d'} \cdot p \cdot \sqrt{\ln(e/p)})
$$
\end{corollary}
\begin{proof}
We replace the tree $T$ with a probabilistic tree $T'$ of depth $d''$ with the same acceptance probability, and the same  coefficients $p_v$  and $\hat{A_v}(\{\next(v)\})$ for $v\in (V_{d'} \cup V_{d'+1} \cup \ldots \cup V_{d''-1})$.
To do so, whenever we reach a node $v'$ in layer $d''$, we flip a coin and accept with probability $\Pr[A_{v'}(x)=1]$. This yields a randomized decision tree with the same parameters mentioned above.
We note that since $p\cdot \sqrt{\ln(e/p)}$ is concave for $p\in[0,1]$, it suffices to prove the bound for any deterministic decision tree of depth $d''$.
Thus, for the rest of the proof let us assume that $d''=d$.

Now, for any fixed vertex $v$ in the $d'$-th layer, $V_{d'}$, we use Corollary~\ref{cor:5.2} to bound the contribution of all descendants of $v$.
Denote by $q_v =\Pr[A_v(x)=1]$.
We get
\begin{align*}
\sum_{v': v\mapsto v'}  p_{v'} \cdot |\hat{A_{v'}}(\{\next(v')\})| &=  p_v \cdot \sum_{v': v\mapsto v'} \Pr[v\mapsto v'] \cdot |\hat{A_{v'}}(\{\next(v)\})|\\
&= p_v \cdot O(\sqrt{d''-d'} \cdot q_v \cdot \sqrt{\ln(e/q_v)}).\tag{Corollary~\ref{cor:5.2} on the subtree rooted at $v$}
\end{align*}
Summing over all $v\in V_i$ we get
\begin{align*}
\sum_{v' \in (V_{d'} \cup V_{d'+1} \cup \ldots \cup V_{d''-1}) } p_{v'} \cdot |\hat{A_{v'}}(\{\next(v')\})|
&\le \sum_{v\in V_{d'}} p_v \cdot O(\sqrt{d''-d'} \cdot q_v \cdot \sqrt{\ln(e/q_v)})\\
& \le O(\sqrt{d''-d'} \cdot p \cdot \sqrt{\ln(e/p)})
\end{align*}
where the last inequality follows from the fact that $x\cdot \sqrt{\ln(e/x)}$ is concave in $[0,1]$ and the expectation of $q_v$ equals $p$.  This completes the proof.
\end{proof}

Next, we use Corollary~\ref{cor:5.3} to get a bound on all  levels of the Fourier transform.
\begin{theorem}[Level-$\l$ Inequality for Decision Trees]\label{thm:main-decision-trees}
	Let $T$ be a decision tree of depth $d$ computing a Boolean function $f:\{-1,1\}^n \to \{0,1\}$  with $p=\Pr[f(x)=1]$. Then, for $\l\le d$,
	$$
	\sum_{S:|S|=\l} |\hat{f}(S)| \le\sqrt{c^\l \cdot \binom{d}{\l}} \cdot p \cdot \prod_{i=0}^{\l-1}\sqrt{ \log(4n^{i}/p)}
	$$
	where $c$ is a universal constant. 
	\end{theorem}	
\begin{proof}We wish to bound $\sum_{|S|=\l} |\hat{f}(S)|$. Note however, that this is equivalent to bounding  $\sum_{|S|=\l} \hat{f}(S) \cdot a_S$ for all $\{-1,1\}$ valued vectors of coefficients $\{a_S\}_{S:|S|=\l}$. 
	Fix such a vector of coefficients $\{a_S\}_{S:|S|=\l}$.
Using Eq.~\eqref{eq:decomposition Fourier}, we may write
\begin{align*}
	\sum_{S:|S|=\l} \hat{f}(S) \cdot a_S &= 
	\sum_{j}
	\sum_{v:\next(v)=j} \hat{A_v}(\{j\}) \cdot \sum_{S_1:|S_1|=\l-1} a_{S_1 \cup \{j\}} \cdot \hat{B_v}(S_1)
\end{align*}

Fix two depths $d'$ and $d''$ such that $0 \le d' < d''\le d$. We wish to bound the contribution of all nodes $v$ of depth between $d'$ and $d''$ to the above sum.
To do so, we note that the Fourier coefficients $|\hat{A_v}(\{j\})|$ are integer multiples of $2^{-d}$. For each $j\in[n]$ and $t\in \{1,\ldots,d\}$ we consider all vertices $v$ of depth between $d'$ and $d''$ with $\next(v)=j$, for which the $t$-th bit in binary representation of $|\hat{A_v}(j)|$ equals $1$.
	We take 
	$$
	V^{+}_{j, t} = \{v: depth(v)\in [d',d''), \next(v)=j, \hat{A_v}(j)>0, |\hat{A_v}(j)|_{t} = 1\},
	$$ 
	and
	$$
	V^{-}_{j, t} = \{v:  depth(v)\in [d',d''), \next(v)=j, \hat{A_v}(j)<0, |\hat{A_v}(j)|_{t} = 1\}.
	$$
	Let 
	$$B^{+}_{j,t}(x) = \sum_{v\in V^{+}_{j,t}}B_v(x)\qquad \text{and}\qquad B^{-}_{j,t}(x) = \sum_{v\in V^{-}_{j,t}}B_v(x).$$
	Note that both $B^{+}_{j,,t}$ and $B^{-}_{j,t}$ can be computed by decision trees of depth at most $d''-1$.
	Thus, by the induction hypothesis, we have that 
	$$\left|\sum_{S_1:|S_1|=\l-1} a_{S_1 \cup \{j\}} \cdot \hat{B^{+}_{j,t}}(S_1)\right| 
	\le \sqrt{c^{\l-1} \cdot \binom{d''-1}{\l-1}} \cdot 
	\Pr[B^+_{j,t}=1] \cdot \prod_{i=0}^{\l-2}\sqrt{ \ln(e\cdot n^{i}/\Pr[B^{+}_{j,t}])}.$$
	and similarly for $B^{-}_{j,t}$.
	Now,
		\begin{align*}
		&\sum_{j} \sum_{\substack{v:\next(v)=j,\\ depth(v)\in[d',d'')}} \hat{A_v}(j) \cdot 
		\sum_{|S_1|=\l-1} a_{S_1 \cup \{j\}} \cdot \hat{B_v}(S_1)\\
		&=\sum_{j} \sum_{t} 2^{-t} \cdot 	
		\sum_{S_1}  a_{S_1 \cup \{j\}} \cdot 
		\left(\sum_{v\in V^{+}_{j,t}} \hat{B_v}(S_1) 
		-  \sum_{v\in V^{-}_{j,t}} \hat{B_v}(S_1)\right)\\
		&=\sum_{j} \sum_{t} 2^{-t} \cdot 	
		\sum_{S_1}  a_{S_1 \cup \{j\}} \cdot 
		\left(\hat{B^+_{j,t}}(S_1) 
		-  \hat{B^-_{j,t}}(S_1) \right)\\
		&\le \sum_{j} \sum_{t} 2^{-t} \cdot \sqrt{c^{\l-1} \cdot \binom{d''-1}{\l-1}}\cdot \left(
		\begin{aligned} 
		\Pr[B^{+}_{j,t}]\cdot \prod_{i=0}^{\l-2}\sqrt{ \ln(e\cdot n^{i}/\Pr[B^{+}_{j,t}])} \;\; +\;\;\\
		\Pr[B^{-}_{j,t}] \cdot \prod_{i=0}^{\l-2}\sqrt{ \ln(e\cdot n^{i}/\Pr[B^{-}_{j,t}])} \end{aligned}
		\right)
	\end{align*}
Since $x\cdot \prod_{i=0}^{\l-2}\sqrt{ \ln(e \cdot n^{i}/x)}$ is monotone increasing for $x\in [0,1]$,	the contribution of $(j,t)$ with $\Pr[B^+_{j,t}]\le \frac{p}{n}$ is at most $\sqrt{c^{\l-1} \cdot \binom{d''-1}{\l-1}} \cdot p \cdot  \prod_{i=0}^{\l-2}\sqrt{ \ln( e\cdot n^{i+1}/p)}$.
	Similarly, the contribution of $(j,t)$ with $\Pr[B^-_{j,t}]\le \frac{p}{n}$ is at most $\sqrt{c^{\l-1} \cdot \binom{d''-1}{\l-1}} \cdot p \cdot  \prod_{i=0}^{\l-2}\sqrt{ \ln(e \cdot n^{i+1}/p)}$.
	For $(j,t)$ and $sg \in \{+,-\}$ with $\Pr[B^{sg}_{j,t}] \ge \frac{p}{n}$, we have that $\prod_{i=0}^{\l-2}\sqrt{ \ln(e \cdot n^{i}/\Pr[B^{sg}_{j,t}])} \le \prod_{i=0}^{\l-2}\sqrt{ \ln(e \cdot n^{i+1}/p)}$ so we can conclude that
	\begin{align*}
		\qquad&\sum_{j} \sum_{\substack{v:\next(v)=j,\\ depth(v)\in[d',d'')}} \hat{A_v}(j) \cdot 
		\sum_{|S_1|=(\l-1)} a_{S_1 \cup \{j\}} \cdot \hat{B_v}(S_1)\\
		&\le \sqrt{c^{\l-1} \cdot \binom{d''-1}{\l-1}} \cdot \prod_{i=0}^{\l-2}\sqrt{ \ln(e \cdot n^{i+1}/p)} \cdot \left(p + p  + \sum_{j}\sum_{t} 2^{-t}(\Pr[B^{+}_{j,t}=1]+\Pr[B^{-}_{j,t}=1])\right)\\
		&= \sqrt{c^{\l-1} \cdot \binom{d''-1}{\l-1}}\cdot \prod_{i=0}^{\l-2}\sqrt{ \ln(e \cdot n^{i+1}/p)}  \cdot \left(2p + \sum_{j}\sum_{\substack{v:\next(v)=j,\\ depth(v)\in[d',d'')}} |\hat{A_v}(j)| \cdot p_v\right)
		\\		&\le \sqrt{c^{\l-1} \cdot \binom{d''-1}{\l-1}}\cdot \prod_{i=0}^{\l-2}\sqrt{ \ln(e \cdot n^{i+1}/p)} \cdot \left(2p + c' \cdot \sqrt{d''-d'}\cdot p\cdot \sqrt{\ln(e/p)}\right)\tag{Corollary~\ref{cor:5.3}}\\
		&\le \sqrt{c^{\l-1} \cdot \binom{d''-1}{\l-1}}\cdot  \prod_{i=0}^{\l-1}\sqrt{ \ln(e \cdot n^{i}/p)} \cdot (c'+2) \cdot \sqrt{d''-d'}\cdot p\;.\end{align*}
To finish the proof we define the sequence of depths $d_0, d_1, d_2, \ldots d_\l$ by $d_i = \lfloor{(d\cdot i)/\l\rfloor}$ for $i=0,\ldots, \l$.
We have that 
\begin{align*}
\sum_{S:|S|=\l} \hat{f}(S)\cdot a_S &= \sum_{i=1}^{\l}
		 \sum_{j} \sum_{\substack{v:\next(v)=j,\\ depth(v)\in[d_{i-1},d_i)}} \hat{A_v}(j) \cdot 
		\sum_{|S_1|=(\l-1)} a_{S_1 \cup \{j\}} \cdot \hat{B_v}(S_1)\\
		&\le 
		\sum_{i=1}^{\l} \sqrt{c^{\l-1} \cdot \binom{d_i-1}{\l-1}} \prod_{i=0}^{\l-1}\sqrt{ \ln(e \cdot n^{i}/p)} \cdot (c'+2) \cdot \sqrt{d_i-d_{i-1}}\cdot p \\
		&\le  \sqrt{c^{\l-1}} \cdot \prod_{i=0}^{\l-1}\sqrt{ \ln(e \cdot n^{i}/p)} \cdot (c'+2) \cdot p \cdot \sum_{i=1}^{\l} \sqrt{\binom{d_i-1}{\l-1} \cdot (d_i-d_{i-1})}
\end{align*}
where the only term that depends on the sequence of depths is the sum $\sum_{i=1}^{\l} \sqrt{\binom{d_i-1}{\l-1} \cdot (d_i-d_{i-1})}$.
A small calculation shows that 
\begin{align*}\sum_{i=1}^{\l} \sqrt{\binom{d_i-1}{\l-1} \cdot (d_i-d_{i-1})} &\le 
\sum_{i=1}^{\l} 
\sqrt{\binom{d-1}{\l-1} \cdot \left(\frac{d_i}{d}\right)^{\l-1}\cdot  \left(1+\frac{d}{\l}\right)} \\&\le \sqrt{2\cdot \binom{d}{\l}}\cdot \sum_{i=1}^{\l}\sqrt{i/\l}^{\l-1} \le O\left(\sqrt{\binom{d}{\l}}\right),
\end{align*}
which completes the proof provided that $c$ is a big enough constant.
\end{proof}

\begin{corollary}
	For any decision tree $f$ of depth $d$ on $n$ variables, and every $\l\le d$, we have
	$$
	\sum_{S:|S|=\l} |\hat{f}(S)| \le \sqrt{d^{\l}\cdot  O(\log n)^{\l-1}} \;,
	$$
\end{corollary}

\begin{proof}
	Using Theorem~\ref{thm:main-decision-trees} and relying on the monotonicity of $x\cdot \prod_{i=0}^{\l-1}\sqrt{\ln(e\cdot n^i/x)}$ in $[0,1]$ we get 
	\begin{align*}\sum_{S:|S|=\l} |\hat{f}(S)| &\le \sqrt{c^{\l}\cdot \binom{d}{\l}}\cdot \prod_{i=0}^{\l-1}\sqrt{\ln(e\cdot n^i)}\\
	&\le \sqrt{O(d/\l)^\l \cdot (\l-1)! \cdot \ln^{\l-1} (e\cdot n)}\\
	&\le \sqrt{d^\l \cdot O(\log n)^{\l-1}}\;.\qedhere\end{align*}\end{proof}

For large $\l$ (namely $\l = \Omega(\sqrt{d/\log n})$) we achieve better bounds by a much simpler argument.
\begin{claim}\label{claim:d choose l}
	For any decision tree $f$ of depth $d$, and every $\l\le d$, we have
	$$
	\sum_{S:|S|=\l} |\hat{f}(S)| \le \binom{d}{\l}\;.
	$$
\end{claim}

\begin{proof}
As in the proof of Theorem~\ref{thm:base},
	for every leaf $\lambda$, we associate a vector $v_{\lambda} \in \{-1,0,1\}^n$. We have $(v_{\lambda})_i = 1$ if the $i$-th bit was queried on the path to $v$ and equaled $1$,  $(v_{\lambda})_i = -1$ if the $i$-th bit was queried on the path to $v$ and equaled $-1$,  and $(v_{\lambda})_i = 0$  otherwise.

For every fixed set $S\subseteq [n]$ the Fourier coefficient $\hat{f}(S)$ equals
$$
\hat{f}(S) = \E_{x}\left[f(x) \cdot \prod_{i\in S} x_i\right] = \E_{\lambda}\left[f(\lambda) \cdot \E_{x}\bigg[\prod_{i\in S}x_i \;|\;\text{$x$ leads to $\lambda$}\bigg]\right] = 
\E_{\lambda}\left[f(\lambda) \cdot \prod_{i\in S} (v_{\lambda})_i\right]
$$
Thus,
$$
\sum_{S:|S|=\l}{|\hat{f}(S)|} = \sum_{S:|S|=\l} \left|\E_{\lambda}\bigg[f(\lambda) \cdot \prod_{i\in S} (v_{\lambda})_i\bigg]\right| \le 
\sum_{S:|S|=\l}\E_{\lambda}\left[\bigg|\prod_{i\in S} (v_{\lambda})_i\bigg|\right]
=\E_{\lambda}\left[\sum_{S:|S|=\l} \bigg|\prod_{i\in S} (v_{\lambda})_i\bigg|\right]
$$
Since every leaf $\lambda$ has at most $d$ nonzero coordinates in $v_{\lambda}$, there are at most $\binom{d}{\l}$ of the subsets $S$ of size $\l$ with $\left|\prod_{i\in S} (v_{\lambda})_i\right|=1$.
This shows that
$\sum_{S:|S|=\l}{|\hat{f}(S)|} \le \binom{d}{\l}.$
\end{proof}
\subsection{Lower Bounds on the $L_{1,\l}$ of depth-$d$ decision trees}\label{sec:examples}
We present examples demonstrating the tightness of our bounds on the $L_{1,\ell}(\cdot)$ of depth-$d$ decision trees, for small $\ell$. In addition, our latter two examples show that one cannot extend  the bound $L_{1,\ell}(f) \le \sqrt{\binom{d}{\ell}}$ from the non-adaptive case (i.e., $d$-juntas) to the adaptive case (i.e., depth-$d$ decision trees). That is, we show  that one must incur a multiplicative factor of roughly $\sqrt{(\log n)^{\ell-1}}$ going from the non-adaptive to the adaptive case.
\begin{example}
	The Majority of $d$ variables can be computed by a depth $d$ decision tree. This function has $L_{1,\l}(\mathsf{MAJ}_d) = \frac{1}{\poly(\l)} \cdot \sqrt{\binom{d}{\l}}$ for odd $\l$ (cf.~\cite[Chapter~5.3]{OdonnellBook}).
\end{example}
\begin{example}
	The address function on $n = 2^{d}+d$ variables, denoted by $\mathsf{Add_{d}}$ can be computed by a depth $d+1$ decision tree.
	In the address function we divide the input to two parts: the first $d$ bits $(x_{1}, \ldots, x_{d})$, called the ``index'', and the latter $2^{d}$ bits $(y_1, \ldots, y_{2^d})$, called the ``array''. We treat  $x$  as representing an index between $1$ and $2^{d}$ that points to the array, and return the coordinate $y_x$.
	It is easy to see that $L_{1,\l}(\mathsf{Add}_{d})=\binom{d}{\l-1}$ exactly. 
	This may seem to rule out any significant improvement over the simple upper bound given in Claim~\ref{claim:d choose l}, namely, $\binom{d}{\ell}$. Note, however, that in the address function $d=\lfloor{\log n\rfloor}$, so in fact this example is consistent with an asymptotic behavior of $\sqrt{\binom{d}{\l}\cdot O(\log n)^{\l-2}}$.
\end{example}
\begin{example}
Let $D = 2d$. Take the address function $\mathsf{Add}_d$ and replace each variable in the array part with the Majority function on $d$ distinct new variables. Denote the resulting function by $f$.
Then, $f$ is a function on $d + 2^d\cdot d$ variables, which can be computed by a depth-$2d$ decision tree, and has $L_{1,\l}(f) \ge \Omega(\sqrt{d} \cdot \binom{d}{\l-1})$. Again, as in the previous example $d = \Theta(\log n)$ so the behavior is consistent with a $\sqrt{\binom{d}{\l}\cdot O(\log n)^{\l-1}}$
bound.
\end{example}

\section{Open Questions}
We would like to highlight several open questions that were mentioned throughout the manuscript. The first is stated as Conjecture~\ref{conj}, namely what are the tight bounds on the $L_{1,\ell}(\cdot)$ of shallow decision trees?
Our conjectured bounds would imply a $\tilde{\Omega}(N^{1-1/k})$ lower bound on the randomized query complexity of the $k$-fold Rorrelation problem, which would be tight due to the upper bounds in \cite{AA15}.

The second question asks whether one can use tools from stochastic calculus to analyze $\E[f(\cDUk)] - \E[f(\cU_k)]$. Such analysis could potentially rely only on level-$k$ bounds on the Fourier spectrum of $f$ (and its restrictions) as done in \cite{RT19} for $2$-fold Forrelation.

The third question asks whether one can exhibit an explicit family of orthogonal matrices $\{U_{N}\}_{N}$ for infinitely many input lengths $N$, such that (1) $U_N$ can be implemented by $\polylog(N)$ size quantum circuits and (2) $U_N$  are good orthogonal matrices as in Definition~\ref{def:good}.
Our current separation uses random orthogonal matrices, that are non-explicit, and cannot be implemented efficiently.

\subsection*{Acknowledgements}
I would like to thank 
Scott Aaronson, 
Rotem Arnon-Friedman, 
Adam Bouland,
Uma Girish,
Tarun Kathuria, 
Chinmay Nirkhe,
Prasad Raghavendra,
Ran Raz, 
Li-Yang Tan,
and 
Umesh Vazirani for very helpful discussions.

\newcommand{\etalchar}[1]{$^{#1}$}

\appendix
\section{Proof of Claim~\ref{claim:AA for Rorrelation}}\label{app:AA}
\begin{proof}
Suppose that $N$ is a power of two, and denote by $n = \log_2(N)$.
We follow the algorithm suggested by Aaronson and Ambainis \cite{AA15} but replace the Hadamard transform (in some of the places) with the orthogonal transform $\U$.

Let $H^{\otimes n}$ be the Hadamard transform on $\R^N$, 
let $\U$ be the orthogonal transform on $\R^N$ from the definition of $k$-fold Rorrelation. For $i\in [k]$ let $U_{z^{(i)}}$ be the query transformation that maps $\left\vert
j\right\rangle $ to $z^{(i)}_j  \cdot \left\vert j\right\rangle $ for all $j\in [N]$ (recall that $z^{(i)} \in \{-1,1\}^N$, thus this is a unitary transformation).

We start with the initial state $\left\vert 0\right\rangle^{\otimes n}$, in addition to a control qubit in the state $\left\vert +\right\rangle =\frac
{\left\vert 0\right\rangle +\left\vert 1\right\rangle }{\sqrt{2}}$. 
Then, conditioned on the control qubit being $\left\vert 0\right\rangle $, we apply
the following sequence of operations to the initial state:
\[
H^{\otimes n} \rightarrow U_{z^{(1)}}\rightarrow
\U^{\T} \rightarrow U_{z^{(2)}}\rightarrow
\cdots 
\U^{\T} \rightarrow U_{z^{(\left\lceil k/2\right\rceil)}}\rightarrow 
\U^{\T}
\]
Meanwhile, conditioned on the control qubit being $\left\vert 1\right\rangle
$, we apply the following sequence of operations:
\[
H^{\otimes n} \rightarrow U_{z^{(k)}}\rightarrow
\U \rightarrow U_{z^{(k-1)}}\rightarrow
\cdots 
\U \rightarrow U_{z^{(\left\lceil k/2\right\rceil+1)}}
\]
Finally, we measure the control qubit in the $\{\ket{+},\ket{-}\}$ basis, and
accept if and only if we find it in the state
$\ket{+} $.

It remains to show that the acceptance probability equals $\frac{1+\phi_{\U}(z^{(1)}, \ldots, z^{(k)})}{2}$, as we do next:
\begin{itemize}
	\item Conditioned on the control qubit being $\left\vert 0\right\rangle
$, the quantum state can be written in vector form as 
$$
\vec{a} = \U^{\T} \cdot U_{z^{(\left\lceil k/2\right\rceil)}} \cdot \U^{\T} \cdots \U^{\T} \cdot U_{z^{(2)}} \cdot \U^{\T} \cdot \vec{v}_{z^{(1)}}
$$
where $\vec{v}_{z^{(1)}}$ is the $N$-dimensional vector with $i$-th entry $\frac{1}{\sqrt{N}} \cdot z^{(1)}_i$.
\item Conditioned on the control qubit being $\left\vert 1\right\rangle
$, the quantum state can be written in vector form as 
$$
\vec{b} = U_{z^{(\left\lceil k/2\right\rceil+1)}} \cdot \U \cdots \U \cdot U_{z^{(k-1)}} \cdot \U \cdot \vec{v}_{z^{(k)}}
$$
where $\vec{v}_{z^{(k)}}$ is the $N$-dimensional vector with $i$-th entry $\frac{1}{\sqrt{N}} \cdot z^{(k)}_i$.
\end{itemize}
Overall, our combined quantum state is  
\begin{align*}
&\frac{1}{\sqrt{2}}\sum_{i=1}^{N}a_i \cdot \ket{i}\ket{0} + 
\frac{1}{\sqrt{2}}\sum_{i=1}^{N}b_i \cdot \ket{i}\ket{1} = 
\frac{1}{2}\sum_{i=1}^{N}(a_i+b_i) \cdot \ket{i}\ket{+} + 
\frac{1}{2}\sum_{i=1}^{N}(a_i-b_i) \cdot \ket{i}\ket{-}\end{align*}
Measuring the control bit in the $\{\ket{+},\ket{-}\}$ basis yields $\ket{+}$ with probability 
$$\frac{1}{4}\sum_{i=1}^{N} (a_i+b_i)^2 = \frac{1}{4}\left(\|a\|_2^2+\|b\|_2^2 + 2\langle{\vec{a}, \vec{b}\rangle} \right)$$
Observe that both $\vec{a}$ and $\vec{b}$ are unit vectors, since they are generated by applying orthogonal matrices to the unit vectors $\vec{v}_{z^{(1)}}$ and $\vec{v}_{z^{(k)}}$, correspondingly.
Furthermore, observe that 
\begin{align*}
\langle{\vec{a}, \vec{b}\rangle} 
&= \vec{a}^{\T} \cdot \vec{b}\\  &= 
\left(
\vec{v}_{z^{(1)}}^{\T}
\cdot  \U 
\cdot U_{z^{(2)}}
\cdot \U \cdots \U \cdot U_{z^{(\lceil{k/2\rceil})}} \cdot \U
\right)\cdot \left(U_{z^{(\left\lceil k/2\right\rceil+1)}} \cdot \U \cdots \U \cdot U_{z^{(k-1)}} \cdot \U \cdot \vec{v}_{z^{(k)}}\right)
 \\&= \frac{1}{N}\cdot \sum_{i_1, \ldots, i_k} z^{(1)}_{i_1} \cdot \U_{i_1,i_2} \cdot z^{(2)}_{i_2} \cdot  \U_{i_2,i_3}  \cdots \U_{i_{k-1},i_k} \cdot z^{(k)}_{i_k}\\
 &=\phi_{\U}(z^{(1)}, \ldots, z^{(k)}).
\end{align*}
Thus, overall we got that the algorithm's acceptance probability is 
\[
\frac{1+1+2 \langle{\vec{a}, \vec{b}\rangle} }{4}=\frac{1+\phi_{\U}(z^{(1)}, \ldots, z^{(k)})}{2}.\qedhere
\]\end{proof}

\section{Pseudorandomness of $\cDUk$ Based on Conjecture~\ref{conj}}
\label{app:B}
\begin{theorem}\label{thm:main-conj}
Assume that Conjecture~\ref{conj} holds.
Suppose that $\U$ is a good orthogonal matrix and $\cDUk$ is defined with respect to $\U$. 
Let $F$ be a {randomized} decision tree of depth $d$ over $kN$ variables.
	Then, $$\E[F(\cU_k) - F(\cDUk)] \le 	O\left(\frac{d \cdot  (\log kN)^{2-1/k}}{ N^{1-1/k}}\right)^{k/2}.$$
\end{theorem}
In other words, assuming the conjecture, for good $\U$, $\cDUk$ is pseudorandom against any depth $N^{1-1/k}/\polylog(kN)$ randomized decision tree. 

\begin{proof}
Without loss of generality $\frac{d \cdot  (\log kN)^{2-1/k}}{ N^{1-1/k}} \ll 1$ as otherwise the claim is trivial. We have
\begin{align*}
|\E[F(\cU_k) - F(\cDUk)]|&= 
|\sum_{S\neq \emptyset} {\hat{F}(S) \cdot \hat{\cDUk}(S)}| \\
&\le \sum_{\ell=k}^{kN} \sum_{|S|=\ell}{|\hat{F}(S) \cdot \hat{\cDUk}(S)|}\\		
&\le \sum_{\ell=k}^{kN} \left(\max_{|S|=\ell} |\hat{\cDUk}(S)|\right) \cdot \sum_{|S|=\ell}{|\hat{F}(S)|}\\
&\le \sum_{\ell=k}^{kN}	\left(\frac{c \cdot \ell \cdot \ln N}{N}\right)^{\ell\cdot (1-1/k)/2} \cdot \sqrt{\binom{d}{\ell}\cdot O(\log kN)^{\ell}} \\
&\le \sum_{\ell=k}^{kN}	O\left(\left(\frac{\ell \cdot \ln N}{N}\right)^{1-1/k} \cdot \frac{d\cdot \log (kN)}{\ell} \right)^{\ell/2}\\
&\le \sum_{\ell=k}^{kN}	O\left(\frac{d \cdot (\log kN)^{2-1/k}}{N^{1-1/k}}\right)^{\ell/2} \\
&\le O\left(\frac{d \cdot (\log kN)^{2-1/k}}{N^{1-1/k}}\right)^{k/2}\;.
\qedhere
\end{align*}
\end{proof}

Based on that, the lower bound on the randomized decision tree complexity of $k$-fold Rorrelation can be improved to $N^{1-1/k}/(k\cdot \polylog(N))$, by simply replacing Theorem~\ref{thm:main} with Theorem~\ref{thm:main-conj} in the proof of Claim~\ref{claim:hard-dist}.
\end{document}